\documentclass[12pt, draftclsnofoot, onecolumn]{IEEEtran}

\usepackage[T1]{fontenc}
\usepackage{amsmath}
\usepackage[cmintegrals]{newtxmath}
\usepackage{dsfont} 

\usepackage{amsthm} 
\usepackage[lined, ruled, linesnumbered]{algorithm2e}

\usepackage[square, comma, sort&compress, numbers]{natbib} 
\usepackage{multirow}
\usepackage{xcolor}
\usepackage{booktabs} 
\usepackage{colortbl} 
\usepackage{graphicx}
\graphicspath{{./figs/}}

\usepackage{csquotes}

\newtheorem{lemma}{Lemma}
\newtheorem{theorem}{Theorem}
\newtheorem{corollary}{Corollary} 

\usepackage[font=scriptsize]{caption}
\DeclareCaptionLabelFormat{algnonumber}{Algorithm}
\newenvironment{framework}[1][htb]
{

	\begin{algorithm}[#1]%
	}{\end{algorithm}}

\usepackage{hyperref}
\newcommand\myeq{\mathrel{\overset{\makebox[0pt]{\mbox{\normalfont\tiny\sffamily def}}}{=}}}
\newcommand{\Expect}{{\rm I\kern-.3em E}}
\newcommand{\Probability}{{\rm I\kern-.3em Pr}}

\newcommand{\problemName}{spectrum monitoring problem}

\DeclareMathOperator*{\argmax}{argmax} 

\newcommand{\I}{{\rm I}}
\newcommand{\II}{{\rm I\kern-0.2ex I}}
\newcommand{\III}{{\rm I\kern-0.2ex I\kern-0.2ex I}}
\newcommand{\IV}{{\rm I\kern-0.2ex V}}

\newcommand{\alg}{\mbox{\rm SpecWatch}}

\newcommand{\algsfull}{{\algB} and {\algC}}
\newcommand{\algA}{\mbox{\rm SpecWatch-{\I}}}
\newcommand{\algAabr}{{\I}}

\newcommand{\algB}{\mbox{\rm SpecWatch-{\II}}}
\newcommand{\algBabr}{{\II}}

\newcommand{\algC}{\mbox{\rm SpecWatch-{\III}}}
\newcommand{\algCabr}{{\III}}

\begin{document}
	
	\title{SpecWatch: A Framework for Adversarial Spectrum Monitoring with Unknown Statistics}
	
	\author{
	Ming~Li,~\IEEEmembership{Member,~IEEE,}
	Dejun~Yang,~\IEEEmembership{Member,~IEEE,}
	Jian~Lin,~\IEEEmembership{Member,~IEEE,}        Ming~Li,~\IEEEmembership{Member,~IEEE,}
	and~Jian~Tang,~\IEEEmembership{Senior Member,~IEEE}
		\vspace{-.3in}
\thanks{This is an extended and enhanced version of the paper [50] that appeared in INFOCOM 2016. This research was supported in part by NSF grants 1443966, 1444059, and 1619728. The information reported here does not reflect the position or the policy of the federal government.}
\thanks{
M. Li, D. Yang and J. Lin are with Colorado School of Mines, Golden, CO 80401.
E-mail: \texttt{\{mili, djyang, jilin\}@mines.edu}.}
\thanks{ M. Li is with the University of Arizona, Tucson, AZ 85721. E-mail: \texttt{lim@email.arizona.edu}.}
\thanks{J. Tang is with Syracuse University, Syracuse, NY 13210. E-mail: \texttt{jtang02@syr.edu}.}
}
	
	\maketitle
	
	\vspace{-.5in}	
	\begin{abstract}
	In cognitive radio networks (CRNs), dynamic spectrum access has been proposed to improve the spectrum utilization, but it also generates spectrum misuse problems. 
	One common solution to these problems is to deploy monitors to detect misbehaviors on certain channel.
	However, in multi-channel CRNs, it is very costly to deploy monitors on every channel. 
	With a limited number of monitors, we have to decide which channels to monitor.
	In addition, we need to determine how long to monitor each channel and in which order to monitor, because switching channels incurs costs. 
	Moreover, the information about the misuse behavior is not available a priori. 
	To answer those questions, we model the {\problemName} as an adversarial multi-armed bandit problem with switching costs (MAB-SC), propose an effective framework, and design two online algorithms, {\algsfull}, based on the same framework. 
	%
	%
	%
	To evaluate the algorithms, we use \textit{weak regret},  i.e., the performance difference between the solution of our algorithm and optimal (fixed) solution in hindsight, as the metric.
	We prove that the \textit{expected} weak regret of {\algB}  is $O(T^{2/3})$, where $T$ is the time horizon. Whereas, the \textit{actual} weak regret of {\algC}  is $O(T^{2/3})$ with probability $1-\delta$, for any $\delta \in (0,1)$.
	Both algorithms guarantee the upper bounds matching the lower bound of the general adversarial MAB-SC problem.
	%
	Therefore, they are all asymptotically optimal.
	
\end{abstract}

	
	\IEEEpeerreviewmaketitle
	
	\section{Introduction}

With the proliferation of  wireless devices and applications, demand for access to spectrum has been growing dramatically and is likely to continue to grow in the foreseeable future~\cite{force2002report}. 
However, there is a paradoxical phenomenon that usable radio frequencies are exhausted while much of the licensed spectrum lies idle at any given time and location~\cite{zhao2007survey}.
To improve the radio spectrum utilization efficiency, dynamic spectrum access (DSA) in cognitive radio networks (CRNs) has been proposed as a promising approach.
Among various DSA strategies, opportunistic spectrum access (OSA) based on the hierarchical access model has received much attention recently~\cite{santivanez2006opportunistic, anandkumar2010opportunistic, tekin2011online, xu2012opportunistic,altrad2014opportunistic,yadav2015opportunistic,wang2016jamming,tsiropoulos2016radio}.
This underlay approach achieves spectrum sharing by allowing secondary (unlicensed) users (SUs) to dynamically search and access the spectrum vacancy while limiting the interference perceived by primary (licensed) users (PUs)~\cite{wang2011anti}. 
OSA helps to improve the spectrum utilization but also results in spectrum misuse or abuse problems due to the flexibility of spectrum opportunity.
For example, an SU may intentionally disobey the interference constraints set by the PU; or some greedy SUs may transmit more aggressively in time and frequency to dominate the spectrum sharing, or even emulate the PU to prevent other SUs from sharing.
Through such spectrum access misbehavior, the malicious users (MUs), i.e., the misbehaving SUs, not only harm the spectrum access operations of normal users, but also impede the CRNs to function correctly since there is no incentive to pay for spectrum access~\cite{yang2012enforcing}. 
Thus, spectrum monitoring is necessary and imperative. 

To address the spectrum misuse problem, different trusted infrastructures have been proposed to detect spectrum misuse and punish MUs~\cite{yang2012enforcing,weiss2013enforcement,sorrells2011anomalous,kumar2016frequency}.
In addition, various detection techniques have been designed, including enforcing silence slots~\cite{atia2008spectrum}, publicizing back-off sequences~\cite{kyasanur2003detection,zhang2013countering}, exploiting spatial pattern of signal strength~\cite{liu2012detecting}, measuring detector value~\cite{tang2013selfish}.
There is also a crowdsourcing-based framework named SpecGuard~\cite{Jin2015specguard} which explores dynamic power control at SUs to contain the spectrum permit in physical layer signals. 
%
	Another crowdsourced enforcement framework~\cite{dutta2016see}
	improves the probability of detection while reducing the likelihood of false positives for spectrum misuses and it can detect misuses caused by mobile users.
%
Moreover, applying big data analysis and machine learning to cloud-based radio access networks also provide an appropriate approach to enable long-term spectrum monitoring~\cite{baltiiski2016long}.
However, all these works assume that all channels can be monitored at the same time.
%
In this paper, we consider spectrum monitoring in multi-channel CRNs with limited monitoring resource. 
In particular, we deploy one monitor with multiple radios where each radio is in charge of one channel. 
The main problem is the selection of channels since monitoring all channels simultaneously is  very energy-consuming and impractical. 
It is challenging  because the information of MUs is unknown a priori.
In addition,  switching costs caused by changing channels must be considered.
To solve this problem, we formulate it as an adversarial (non-stochastic) multi-armed bandit (MAB) problem~\cite{auer2002nonstochastic} with switching costs.
We then propose an effective spectrum monitoring framework and design two online algorithms based on it. 

In summary, we contribute in the following aspects:

\begin{enumerate}
	\item
	We study the adversarial spectrum monitoring problem with unknown statistics in multi-channel CRNs, while considering the switching cost. We model this problem as an adversarial MAB-SC problem.
	\item
	We propose an online spectrum monitoring framework, {\alg}.
	Based on this framework, we design two spectrum monitoring algorithms, {\algsfull},
%
	which differ in the way of calculating strategy probabilities and updating strategy weights.
	 Our algorithms guarantee the proved performance under any type of adversary settings. In addition, they can work with any spectrum misuse detection techniques in the current literature.
	\item
	We prove that the expected weak regret of  {\algB} is $O(T^{2/3})$, which matches the lower bound in~\cite{dekel2014bandits}. Therefore, {\algB} is asymptotically optimal. Note that the expected value of normalized weak regret is guaranteed to be $O(1/T^{1/3})$, which converges to 0 as time horizon $T$ approaches to $\infty$. 
	\item
	{\algC}  select channels more strategically and explore all channels more efficiently. We prove that this algorithm guarantees the actual weak regret to be  $O(T^{2/3})$, which is asymptotically optimal as well, with probability $1 - \delta$, for any $\delta \in (0,1)$.	
\end{enumerate}

%
%
\section{Related Work}
\subsection{Spectrum Monitoring Problem}
%

Spectrum monitoring problem has been formulated as optimization problems with different objectives.
In~\cite{shin2009optimal}, Shin and Bagchi modeled the channel assignment for monitoring wireless mesh networks as a maximum-coverage problem with group budget constraints. 
They then extended it to the model where monitors may make errors due to poor reception~\cite{shin2013toward}. 
Along the same line, Nguyen et al.~\cite{nguyen2014quality} focus on the weighted version of the problem, where users to be covered have weights. 
To maximize the captured data of interest, Chen et al.~\cite{chen2014efficient} utilized support vector regression to guide monitors to intelligently select channels.
Considering similar objectives, Shin et al.~\cite{shin2012distributed} designed a cost-effective distributed algorithm.
With a different approach, Yan et al.~\cite{yan2014specmonitor} solved the problem by predicting secondary users' access patterns. 
However, we consider a different objective in this paper, which is to capture spectrum misuses. In addition, we assume no information about the malicious users.

The closest works to ours were presented in~\cite{arora2011sequential,zheng2014approximate,le2014sequential,yi2012secondary,xu2014secondary}.
In \cite{arora2011sequential}, Arora and Szepesvari first modeled the {\problemName} as a multi-armed bandit problem (MAB) to monitor the maximum number of active users.
They designed two algorithms to learn sequentially the user activities while making channel assignment decisions.
Observing the above algorithms suffer from high computation cost, Zheng et al.~\cite{zheng2014approximate} traded off between the rate of learning and the computation cost.
They proposed a centralized online approximation algorithm
and show that it incurs sub-linear regret bounds over time and a distributed algorithm with moderate message complexity.
In~\cite{le2014sequential}, Le et al. considered switching costs for the first time and utilized Upper Confident
Bound-based (UCB) policy~\cite{auer2002finite} which enjoys a logarithmic regret bound
in time that depends sublinearly on the number of arms, while its
total switching cost grows in the order of $O(\log (\log T))$. 
Considering a different objective, Yi et al.~\cite{yi2012secondary} used UCB to capture as much as interested user data.

However, these works used the stochastic MAB model, where the rewards for playing each arm are generated independently from unknown but fixed distributions.
Our model, in contrast, does not make such assumptions.
The only work considered the similar problem model to ours is~\cite{xu2014secondary}, where Xu et al. tried to capture packets of target SUs for CRN forensics.
However, they did not provide any algorithm whose actual weak regret can be bounded with a confidence value.

\subsection{Multi-arm Bandit Problem}
The MAB problem first introduced by Robbins~\cite{robbins1952some} has been extensively studied in the literature. The classical MAB problem models the trade-offs faced by a gambler who aims to maximize his rewards over many turns by exploring different arms of slot machines and to exploit arms which have provided him more rewards than others. 
The gambler has no knowledge about the reward of each arm a priori and only gains knowledge of the arms he has pulled. An MAB algorithm should specify a strategy by which the gambler chooses an arm at each turn. 
The performance of an algorithm is measured in regret, as will be elaborated in Section \ref{sec: model}. 
There are mainly two algorithm families based on different formulations of MAB. 
The upper confidence bounds (UCB) family of algorithms~\cite{lai1985asymptotically} works for stochastic MAB, whose regret can be as small as $O(\ln  T)$ where $T$ is the number of turns.
However, these algorithms are established  with the assumption that there exist fixed (though unknown) probability distributions of different arms to generate rewards, which may not be satisfied in our spectrum monitoring problem and thus not considered by us. 
The other algorithm family is the EXP3 family~\cite{auer2002nonstochastic}  for adversarial MAB. Auer~\cite{auer2002nonstochastic} has studied MAB with no assumption on the rewards distribution and proposed algorithms with regret of $O(T^{1/2})$. 
There are also some algorithms considering both stochastic and adversarial adversary \cite{bubeck2012regret,auer2016algorithm}.
However, switching costs are not considered in all above works. 
In our model, each time the monitor changes its monitoring channels, there are drastic costs in terms of delay, packet loss, and protocol overhead \cite{lin2011Opportunistic}. These costs must be taken into consideration when designing monitoring algorithms. 
Although there exists some work on stochastic multi-armed bandit problem with switching cost (MAB-SC)~\cite{lin2011Opportunistic}, little research has been done on adversarial MAB-SC.
%
Dekel et al.~\cite{dekel2014bandits} proved the lower bound of the regret for adversarial MAB-SC to be $\tilde{\Omega}$ ($T^{2/3} $). In this paper, the upper bound of regret guaranteed by  our algorithms matches this lower bound. Moreover, different from existing works, the strategy for each turn (or timeslot as in our model) is no longer a single arm because we consider a more general case where multiple channels can be monitored at the same time. 
Therefore, none of above algorithms can be directly applied to our problem.
	
%
%
\section{System Model and Problem Statement}\label{sec: model}

We consider a cognitive radio network which adopts a hierarchical access structure with primary users (PUs) and secondary users (SUs). We assume the spectrum is divided into a set $\mathcal{K} = \{1, 2, \ldots, k , \ldots, K\}$ of $K$ channels. 
The total time period is discretized into a set $\mathcal{T} = \{1, 2, \ldots, t , \ldots, T\}$ of $T$ timeslots.
Ideally, SUs seek spectrum opportunities among $K$ channels in a non-intrusive manner. 
However, the malicious users (MUs) may perform unauthorized access or selfish access. 
We consider the scenario where there exists one monitor with $l$ radios and a set $\mathcal{M} = \{1, 2, \ldots, m, \ldots, M\}$ of $M$ MUs.
Note that for the case of  multiple monitors, if there is a central controller, it is equivalent to one monitor with the same number of radios; otherwise, each monitor can execute our algorithms independently. In the latter case, however, the regret may not be bounded.

Since the monitor is equipped with $l$ radios, it can monitor up to $l$ channels at the same time.  
Assume that one radio is tuned to monitor channel $k$, and there are $M_k$ MUs on that channel, then the detection probability of that radio to successfully detect MUs' presence is $p_d(M_k)$, which is dependent on the monitor's hardware and the detection technique. 
Any technique in~\cite{atia2008spectrum,liu2012detecting,zhang2013countering,tang2013selfish,kyasanur2003detection} can be adopted to detect spectrum misuses.
In practice, the detection probability will also be dependent on the presence of PU and other SUs. However, since our algorithms do not require the knowledge of the detection probability, we simplify the notion to $p_d(M_k)$ where it seems that only $M_k$ matters.

Let  $\{0,1, \ldots, l\}^K$ denote the strategy space of the monitor. A \textit{strategy} $s$ is represented as $(a_{s1}, a_{s2}, \ldots, a_{sK} )$, where the value of the $a_{sk}$ is 1 if a radio is assigned to monitor channel $k$, 0 otherwise. Therefore, $\sum_{k \in \mathcal{K}} a_{sk} = l$.
For example,  considering 4 channels and a monitor with 2 radios, strategy $(0, 1, 0, 1)$ indicates that one radio is tuned to monitor channel 2 and the other radio is tuned to monitor channel 4.
For notational simplicity, we will write $k \in s$ instead of $a_{sk} = 1$ to denote that channel $k$ is chosen in strategy $s$. 
Since each radio is assigned one  out of $K$ channels to monitor, and we have $l$ radios in total, the number of strategies is $S = \binom{k}{l}$. The whole strategy set is represented as $\mathcal{S} = \{1, 2, \ldots, s, \ldots, S\}$.  Note that $K$ and $l$ are usually small. For example, the regulated 2.4 GHz band is divided into only 14 channels. The maximum number of radios on each monitor defined by the active IEEE 802.11af standard is set to be 4~\cite{IEEE80211af}.

%
	In this paper, we assume both the monitor and MUs are static, i.e., staying at the same location. Note that when mobility is considered as in the pursuit-evasion problems~\cite{wang2016learning}, we only need to enlarge the strategy space by including the location dimension.
%

At the beginning of timeslot $t \in \mathcal{T}$, the monitor selects only one strategy from the strategy set $\mathcal{S}$, and we denote the chosen strategy as $X_t$. 
We assume the switching cost $c (X_{t-1}, X_t) \in  \left[ 0, 1 \right]$, but our algorithm can be generalized to any range $\left[ \underline{c}, \bar{c} \right]$, $\underline{c} < \bar{c}$ by scaling, where $\underline{c}$ and $\bar{c}$ are the minimum value and the maximum value of the switching cost, respectively. For simplicity, set the switching cost of the first timeslot to be $c (X_{0}, X_1) = c_0$ regardless of what $X_1$ is.
Clearly, $c (X_{t-1}, X_t)=0$ if $X_{t-1}= X_t$.

\noindent \textbf{Threat Model:} At each timeslot $t \in \mathcal{T}$, each MU $m \in \mathcal{M}$ chooses one channel to attack (conduct misuses) according to its attack probability distribution $\mathbf{P}_{t}^{m} = \left\{ P_{t, 1}^{m}, \ldots, P_{t, K}^{m}\right\}$ 
where $P_{t, k}^{m}$ denotes the probability of MU $m$ attacking channel $k$ in timeslot $t$. Since MUs may not attack in some timeslot, $\sum_{k \in \mathcal{K}} P_{t, k}^{m} \leq 1$ for any $m \in \mathcal{M}$ and $t \in \mathcal{T}$. We consider two types of adversary: 
\begin{enumerate}
	\item 
	Oblivious Adversary (Stochastic): The MUs keep their attack patterns regardless of how the monitor work. For any MU $m$, the attack distribution $\mathbf{P}^m_{t}$ remains the same throughout the time horizon. In this paper, we consider three different adversary settings (elaborated in Section~\ref{simulation}): fixed adversary, uniform adversary, and normal adversary.
	
	\item
	Adaptive Adversary (Adversarial): The MUs know every action of the monitor from the beginning to the current timeslot and adjust their strategies accordingly based on any learning algorithms, i.e., the attack distribution might change with time. 
	%
\end{enumerate}
Our framework and algorithms work for both types and the theoretical bounds hold no matter what the adversary type is.

Now we define the reward for the monitor. The \textit{strategy reward} of choosing strategy $s$ in timeslot $t$ is
\begin{equation}
g_{s, t} \myeq 
\begin{cases}
\sum_{k \in s} f_{k,t}  \qquad &\mbox{if } s = X_t,\\
0 \qquad & \mbox{otherwise,}
\end{cases}
\label{eq: strategy reward sum}
\end{equation}
where the \textit{channel reward} $f_{k, t}$ is defined as
\begin{equation}
f_{k, t} \myeq
\begin{cases}
r  \qquad &\mbox{if channel } k \in X_t \text{ and misuse is  detected}, \\
0  \qquad & \mbox{otherwise,}
\end{cases}
\label{eq: reward r}
\end{equation}
where the \textit{unit reward} $r$ is assumed to be scaled and satisfies $rl \leq 1$  for the purpose of mathematical analysis. 
%
Note that the probability of at least one MU being detected on monitored channel $k$ is determined by the number of radios on that channel $M_k$, the detection probability $p_d(M_k)$ and the action of MUs $\{P^m_{t,k}\}^M_{m = 1}$.
We denote the detection probability on channel $k$ by adopting strategy $X_t$ at timeslot $t$ as $P_d\left(a_{X_t k}, p_d(M_k), \{P_{t,k}^{m}\}^M_{m=1}\right)$
which is assumed to be a non-decreasing function in $a_{X_t k}$, $p_d(M_k)$, and $P_{t,k}^m$.
Thus, the channel reward $f_{k, t}$ is $r$ with probability $P_d\left(a_{X_t k}, p_d(M_k), \{P_{t,k}^{m}\}^M_{m=1}\right)$.
Note that the knowledge of this probability is not required.

Assume the monitor follows the strategy sequence $X_1, X_2, \ldots, X_T$  generated by any monitoring Algorithm $\mathbb{A}$.
At the end of timeslot $T$, the \textit{cumulative strategy reward} is 
\begin{equation}
G_\mathbb{A} \myeq \sum_{t = 1}^{T} g_{X_t, t}.
\label{cummu st reward}
\end{equation}
Meanwhile, the monitor incurs \textit{cumulative switching cost} 
\begin{equation}
L_\mathbb{A} \myeq \sum_{t = 1}^{T} c (X_{t-1}, X_t).
\label{cumu sc}
\end{equation}
Thus, the \textit{utility} of the monitor by choosing Algorithm $\mathbb{A}$ is 
\begin{equation}
U_\mathbb{A} = G_\mathbb{A} - L_\mathbb{A}.
\label{eq: utility A}
\end{equation}

To measure the performance of Algorithm $\mathbb{A}$, we use a special case of the worst-case regret, \textit{weak regret}~\cite{auer2002nonstochastic}, as the metric.

The weak regret of Algorithm $\mathbb{A}$ is the difference between the utility by using \textit{best fixed algorithm} and the actual utility by using Algorithm $\mathbb{A}$. 
A fixed algorithm chooses only  one strategy for all timeslots and never switches. The best fixed algorithm is the one resulting in the highest utility among all fixed algorithms. 
The strategy chosen in best fixed algorithm is called the \textit{best strategy}, denoted by $s_{best}$. 
Formally,
$
s_{best}  \myeq  \argmax_{s \in \mathcal{S}}  \left(\sum_{t = 1}^{T} g_{s, t} - c_0\right),
$
and the utility by using the best fixed algorithm is
\begin{equation}
U_{best}  \myeq  G_{best} - L_{best},
\label{eq: utility best}
\end{equation}
where $G_{best} = \max_{s \in \mathcal{S}} \sum_{t = 1}^{T} g_{s, t} $ and $L_{best} = c_0$ since the switch only happens at the first timeslot.
%
Note that the best strategy can only be found in hindsight. 

Now we can define the weak regret of Algorithm $\mathbb{A}$ as
\begin{equation}
\label{eq: weak regret}
R_\mathbb{A} \myeq U_{best} - U_\mathbb{A}.
\end{equation}

\noindent \textbf{Problem Statement:} Given $K$ channels, time horizon $T$, and a monitor with $l$ radios, our objective is to design online spectrum monitoring algorithms such that the weak regret is minimized, in the presence of different adversaries. We make no assumption on the knowledge of the probability functions $p_d(M_k)$ and $P_d\left(a_{X_t k}, p_d(M_k), \{P_{t,k}^{m}\}^M_{m=1}\right)$. In addition, the attack distribution $\mathbf{P}^m_{t}$ and the reward of choosing a strategy are unknown a priori. 

Therefore, any desired algorithm needs to balance not only the trade-off between exploration and exploitation, but also that between strategy rewards and switching costs. This is a very challenging problem.

	\section{Spectrum Monitoring Framework}

In this section, we design SpecWatch, a spectrum monitoring framework, based on the batching version of exponential-weight algorithm for exploration and exploitation, where the idea of batching is inspired by \cite{arora2012online}. 

To control the trade-off between the reward and the switching cost,  we group all the timeslots into consecutive and disjoint batches.
Within each batch, we stick to the same strategy to avoid the switching cost. Between batches, we reselect a strategy to gain higher rewards. 
A smaller batch size may result in larger reward but larger switching cost, while a bigger batch size may result in smaller switching cost but smaller reward.

%
%
\begin{framework}
	\caption{}
	\label{SpecWatch}
	\textbf{Parameter:} {$\tau \in [1, T]$} 
	
	\textbf{Initialization:} {
		$\mathsf{w}_{s,1} \leftarrow 1$ for all $s \in \mathcal{S}$,
		$J \leftarrow \lceil T / \tau \rceil$.
		}
	
	\For {$j \leftarrow 1, \ldots, J$}
	{
		Calculate the strategy probability $\mathsf{p}_{s, j}$ for all $s \in \mathcal{S}$ 
		
		Choose strategy $Z_j \in \mathcal{S}$ randomly accordingly to the probability distribution $\mathsf{p}_{1,j}, \ldots, \mathsf{p}_{S,j}$ and incur switching cost $c (Z_{j-1}, Z_j)$.
		
		Monitor using $Z_j$ for $\tau$ timeslots, i.e., $X_{[j] + i} \leftarrow Z_j$ for $1 \leq i \leq \tau $, and record the reward of each monitored channel, $f_{k, [j] + i}$ for all $k \in Z_j$, $1 \leq i \leq \tau$.
		
		Update strategy weight $\mathsf{w}_{s, j+1}$ for all $s \in \mathcal{S}$.
	}
\end{framework}
The design details of {\alg} are illustrated above. 
We first set the initial weight of each strategy to be 1. 
Given the batch size $\tau$,  the timeslots $1, 2, \ldots, T$ are divided into $J = \lceil T / \tau \rceil$ consecutive and disjoint batches.
Let $[j] = (j-1)\tau$ for $1 \leq j \leq J$. Then the $j$-th batch starts from timeslot $[j] + 1$ and ends at timeslot $[j] + \tau$ as shown in Fig.~\ref{fig:batch}.
At the beginning of each batch, we calculate the strategy probabilities according to strategy weights. Then we randomly select a strategy based on the probability distribution. During the whole batch, the chosen strategy remains the same. At the end of each batch, strategy weights are updated according to the strategy rewards. 

\begin{figure}[ht!] 
	\centering
	\includegraphics[width = \textwidth]{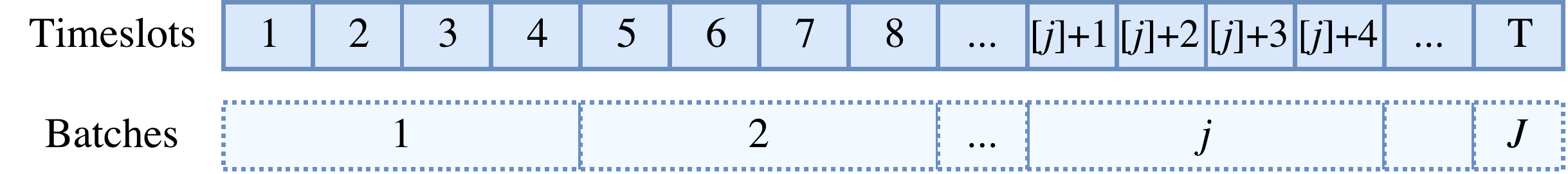}	
	\caption{Example of batching timeslots with batch size $\tau=4$}
	\label{fig:batch}
\end{figure}

%
%

\setlength{\extrarowheight}{1pt}
\begin{table}[h!]
	\centering
	\caption{\textsc{Main Notations}}
	\label{tab:notation}
	\begin{tabular}{ l  l}
		\toprule 
		Notation & Meaning
		\\
		\midrule
		$\mathcal{K}$ 
		&
		set of all channels
		\\
		$\mathcal{T}$ 
		&
		set of all timeslots
		\\
		$\mathcal{M}$
		&
		set of malicious users
		\\
		$\mathcal{S}$
		&
		set of all strategies
		\\
		$\mathcal{C}$ & covering strategy set		
		\\
		$C_k$ & number of strategies in $\mathcal{C}$ that contains channel $k$ 
		\\	
		$l$ 
		&
		number of radios in the monitor
		\\
		$r$
		& unit channel reward  which satisfies $rl \leq 1$
		\\
		$c (X_{t-1}, X_t)$
		&
		switching cost from strategy $X_{t-1}$ to $X_t$
		\\
		$G_\mathbb{A} $
		&
		cumulative strategy rewards of Algorithm $\mathbb{A}$
		\\
		$L_\mathbb{A} $
		&
		cumulative switching costs of Algorithm $\mathbb{A}$
		\\
		$U_\mathbb{A}$
		&
		utility of Algorithm $\mathbb{A}$ 
		\\
		$R_\mathbb{A}$
		&
		weak regret of Algorithm $\mathbb{A}$
		\\
		$\tau$ & parameter to determine the batch size
		\\
		$J$ & number of batches
		\\
		$\gamma$ & parameter to calculate strategy probabilities
		\\
		$\eta$ & parameter to update weights
		\\
		$\beta$ & parameter to calculate channel scores
		\\
		\arrayrulecolor{gray}\midrule
		$f_{k, t}$  &  channel reward of channel $k$ in timeslot $t$
		\\ 
		$g_{s, t}$  & strategy reward of strategy $s$ in timeslot $t$
		\\
		\arrayrulecolor{gray}\midrule
		$\bar{\mathsf{f}}_{k, j}$
		& average channel reward of channel $k$ in batch $j$
		\\ 
		$\bar{\mathsf{g}}_{s, j}$
		& average strategy reward of strategy $s$ in batch $j$
		\\
		$\bar{\mathsf{f}}^{\prime}_{k, j}$
		& average channel score of channel $k$ in batch $j$
		\\ 
		$\bar{\mathsf{g}}^{\prime}_{s, j}$
		& average strategy score of strategy $s$ in batch $j$
		\\
		$\mathsf{h}_{k, j}$ 
		& channel weight of channel $k$ in batch $j$
		\\ 
		$\mathsf{w}_{s, j}$
		& strategy weight of strategy $s$ in batch $j$
		\\ 
		$\mathsf{W}_j $
		& total weight of all strategies in batch $j$
		\\ 
		$\mathsf{q}_{k, j}$
		& channel probability of channel $k$ in batch $j$
		\\
		$\mathsf{p}_{s, j}$
		& strategy probability of strategy $s$ in batch $j$		
		\\
		\bottomrule
	\end{tabular}
\end{table}


In the following sections, we design two effective online spectrum monitoring algorithms, {\algsfull}. These two algorithms are designed based on {\alg} and {\alg}$^{+}$ in our paper~\cite{li2016specwatch}. We rename  {\alg} and {\alg}$^{+}$ to {\algA} and {\algC}, respectively. We introduce {\algB} in this paper and discard {\algA} because {\algB} has better theoretical performance than {\algA}. 

The main notations are summarized in Table~\ref{tab:notation}.

Different algorithms use different equations to calculate the strategy probabilities (Line 3) and update the strategy weights (Line 6).
With carefully chosen parameters, their theoretical performances are shown in Table~\ref{tab: algo perf}. 
In summary, {\algA} and {\algB} bound the \textit{expected} weak regrets to $O(T^{2/3})$; while {\algC} bounds the \textit{actual} weak regret to $O(T^{2/3})$ with user-defined probability. 


\def\arraystretch{2}
\begin{table}[h!]
	\centering
	\caption{\textsc{Algorithm Comparison}}
	\label{tab: algo perf}
	\begin{tabular}{ c c c}
		\toprule 
		Algorithm & Core Functions & Theoretical Bound of Weak Regret
		\\
		\midrule
		\multirow{ 2}{*}{\algA} & $\mathsf{p}_{s, j} = (1-\gamma)\frac{\mathsf{w}_{s, j}}{\mathsf{W}_j} + \frac{\gamma}{S}$ &  
		\multirow{ 2}{*}{$\Expect\left[R_{{\algAabr}}\right] 
			\leq
			3\left((e-1)S\ln S\right)^{\frac{1}{3}} T^{\frac{2}{3}}$}
		\\
		&$\mathsf{w}_{s, j+1} = \mathsf{w}_{s, j} \exp \left( \frac{\gamma}{S} \sum_{k \in s} \frac{\frac{1}{\tau} \sum_{i = 1}^{\tau} f_{k, [j] + i}}{\sum_{s: k \in s}\mathsf{p}_{s,j}}\right) $& 
		\\ 
		\arrayrulecolor{black!30}\midrule
		\multirow{ 2}{*}{{\algB}} & $\mathsf{p}_{s, j} = \frac{\mathsf{w}_{s, j}}{\mathsf{W}_j}$ 
		& 
		\multirow{ 2}{*}{$
		\Expect\left[R_{{\algBabr}}\right] 
		\leq
		3\left(\frac{1}{2} S\ln S\right)^{\frac{1}{3}} T^{\frac{2}{3}}
		$}
	\\
	&$\mathsf{w}_{s, j+1} = \mathsf{w}_{s, j} \exp \left( - \eta \sum_{k \in s} \frac{\mathds{1}_{k \in Z_j}(1/l - \frac{1}{\tau} \sum_{i = 1}^{\tau}  f_{k, [j] + i})}{\sum_{s: k \in s}\mathsf{p}_{s,j}} \right)$&
		\\ 
		\midrule
		\multirow{ 2}{*}{{\algC}}
		  & $\mathsf{p}_{s, j} = 
		(1-\gamma)\frac{\mathsf{w}_{s, j}}{\mathsf{W}_j} + 
		\frac{\gamma}{C}\mathds{1}_{s \in \mathcal{C}}$ 
		&
		\multirow{ 2}{*}{$
		\Probability \left[ 
		R_{\algCabr}
		\leq
		2 \left(4\sqrt{lC\ln S} + 2\sqrt{lK\ln\frac{K}{\delta}}\right)^{\frac{2}{3}} T^{\frac{2}{3}}
		\right] 
		\leq
		1 - \delta
		$}
	\\
	&$\mathsf{w}_{s, j+1} = \mathsf{w}_{s, j} \exp \left(  \eta \sum_{k \in s} 
	\frac{
		\frac{1}{\tau} 
		\sum_{i = 1}^{\tau} f_{k, [j] + i} + \beta}{(1-\gamma)\frac{\sum_{s: k \in s}\mathsf{w}_{s, j}}
		{\mathsf{W}_j} + \frac{\gamma C_k}{C}} \right)$&
		\\
		\arrayrulecolor{black}\bottomrule		
	\end{tabular}
\end{table}	


	\section{Spectrum Monitoring Algorithm with Bounded Expected Weak Regret}
In this section, we design one spectrum monitoring algorithm, {\algB}, whose expected weak regret is theoretically bounded.
	
	
%
%
\subsection{{\algB}}\label{ssec: algB}

\subsubsection{Algorithm Design}
%
Compared to {\algA}, the input parameters of {\algB} are $\tau$ and $\eta$, where $\tau$ determines the batch size and $\eta$ is used to calculate strategy weights; the strategy probabilities in {\algB} are proportion to the strategy weights only, and the strategy weights are updated smaller for the next timeslot.

\textbf{Calculating Strategy Probability $\mathsf{p}_{s, j}$.}
The probability of choosing strategy $s \in \mathcal{S}$ is calculated using
\begin{equation}
\label{eq: strategy prob loss}
\mathsf{p}_{s, j} = \frac{\mathsf{w}_{s, j}}{\mathsf{W}_j} ,
\end{equation}
where $\mathsf{w}_{s, j}$ is the strategy weight of strategy $s$ in batch $j$, and $\mathsf{W}_j = \sum_{s \in \mathcal{S}} \mathsf{w}_{s, j}$. 
Recall that in {\algA}, we use the weighted average of two terms,
$
\mathsf{p}_{s, j} = (1-\gamma)\frac{\mathsf{w}_{s, j}}{\mathsf{W}_j} + \frac{\gamma}{S},
$
where the first is to exploit strategies with good reward history, and the second guarantees the exploration over all strategies. $\gamma$ controls the balance between them. 
Now we only have one term, but the balance still exists. This is because a strategy not chosen before is guaranteed to have a higher weight, due to the way we update the strategy weight as discussed below. 


\textbf{Choosing Strategy $Z_j$ and Monitoring Spectrum.}
We select a strategy $Z_j \in \mathcal{S}$ randomly according to the probabilities calculated above. The monitor keeps using $Z_j$ for all $\tau$ timeslots in batch $j$, i.e., $X_{[j] + i} = Z_j$ for $1 \leq i \leq \tau $. Therefore, the monitor only incurs switching cost $c (Z_{j-1}, Z_j)$ once for the whole batch $j$. 

Depending on the misuse behavior of MUs (discussed in Section~\ref{sec: model}), the monitor receives rewards on monitored channels accordingly. The monitor keeps records of $f_{k, [j] + i}$ for all  $k \in Z_j$ and $1 \leq i \leq \tau$. The strategy reward gained by the monitor is the summation of rewards over all monitored channels.

\textbf{Updating Strategy Weight $w_{s, j+1}$.}
At the end of each batch, we update strategy weights in the following steps.

First, we calculate the average channel reward of each channel $k \in \mathcal{K}$ in batch $j$,
\begin{equation}
\label{eq: average channel reward}
\bar{\mathsf{f}}_{k, j}= \frac{1}{\tau}\sum_{i = 1}^{\tau} f_{k, [j] + i}.
\end{equation}
By \eqref{eq: reward r}, $ \bar{\mathsf{f}}_{k, j} \in [0, r]$. 
We also calculate the probability of choosing channel $k \in Z_j$ by summing up the probabilities of strategies containing that channel,
\begin{equation}
\label{eq: channel probability}
\mathsf{q}_{k, j} = \sum_{s: k \in s} \mathsf{p}_{s, j}.
\end{equation}
where $s: k \in s$ denotes any strategy $s$ containing channel $k$.
Based on~\eqref{eq: average channel reward} and~\eqref{eq: channel probability}, we calculate the \textit{average channel score},
\begin{equation}
\label{eq: virtual channel reward}
\bar{\mathsf{f}}^{\prime}_{k, j} = \frac{1/l -\bar{\mathsf{f}}_{k, j}}{\mathsf{q}_{k,j}}\mathds{1}_{k \in Z_j}, 
\end{equation}
where $\mathds{1}_{k \in Z_j}$ is an indicator function which has the value 1 if $k \in Z_j$; 0 otherwise.
This score is always nonnegative since $\bar{\mathsf{f}}_{k, j} \leq r \leq 1/l$.
Then we update each \textit{channel weight} by 
\begin{equation}
\label{eq: channel weight}
\mathsf{h}_{k, j+1} = \mathsf{h}_{k, j} 
\exp 
\left(
- \eta \bar{\mathsf{f}}^{\prime}_{k, j}
\right),
\end{equation}
where  $\mathsf{h}_{k, 1} = 1$ for all $k \in \mathcal{K}$.
Note that the channel weight is non-increasing from batch to batch. 

Finally, we give the formal definition of \textit{strategy weight}, which is defined as 
\begin{equation}
\label{eq: strategy weight}
\mathsf{w}_{s, j} \myeq  \prod_{k \in s} \mathsf{h}_{k, j}.
\end{equation}
Combining~\eqref{eq: channel weight} and~\eqref{eq: strategy weight}, we can directly update the strategy weight for each $s \in \mathcal{S}$ by
\begin{equation}
\label{eq: weight update}			
\mathsf{w}_{s, j+1} = \mathsf{w}_{s, j} \exp ( - \eta \bar{\mathsf{g}}^{\prime}_{s, j}),
\end{equation}
where $\bar{\mathsf{g}}^{\prime}_{s, j}$ is the \textit{average strategy score} for each  $s \in \mathcal{S}$, i.e.,
\begin{equation}
	\bar{\mathsf{g}}^{\prime}_{s, j} = \sum_{k \in s} \bar{\mathsf{f}}^{\prime}_{k, j}.
\end{equation} 
Note that by combining~\eqref{eq: average channel reward},~\eqref{eq: channel probability} and~\eqref{eq: virtual channel reward}, we can directly calculate $\bar{\mathsf{g}}^{\prime}_{s, j}$ directly by
\begin{equation}
\label{eq: virtual strategy reward}
\bar{\mathsf{g}}^{\prime}_{s, j} 
= \sum_{k \in s} \bar{\mathsf{f}}^{\prime}_{k, j} 
= \sum_{k \in s} \frac{\mathds{1}_{k \in Z_j}(1/l - \bar{\mathsf{f}}_{k, j})}{\mathsf{q}_{k,j}}
= \sum_{k \in s} \frac{\mathds{1}_{k \in Z_j}(1/l - \frac{1}{\tau} \sum_{i = 1}^{\tau} f_{k, [j] + i})}{\sum_{s: k \in s}\mathsf{p}_{s,j}}.
\end{equation}

\noindent \textit{Remark.} We do not update strategy weights based on strategy rewards, but instead calculate channel weights~\eqref{eq: channel weight} first. This is because the rewards of monitored channels provide useful information on those \textit{unchosen strategies} containing these channels.

\subsubsection{Performance Analysis}\label{sssec:SW- analysis}
To analyze the performance of {\algB}, we first bound the difference between the reward gained by {\algB} and that by the best fixed algorithm (Lemma~\ref{lemma3}), and then prove the upper bound of the expected weak regret (Theorem~\ref{thm3}). For a better understanding of Theorem~\ref{thm3}, we present a specific bound obtained by a particular choice of parameters $\eta$ and $\tau$ (Corollary~\ref{coro3}).

Recalling \eqref{eq: utility A}, \eqref{eq: utility best} and \eqref{eq: weak regret}, the   weak regret of {\algB} is
\begin{equation}
R_{{\algBabr}}
=
(G_{best} - L_{best}) - (G_{{\algBabr}} - L_{{\algBabr}}) .
\label{eq: weak regret new}
\end{equation}

Since the best fixed algorithm never switches the strategies, and {\algB} only switches between batches for at most $J$ times,  their cumulative switching costs are
\vspace{-.1in}
\begin{equation*}
L_{best} = c_0 \in [0, 1] \quad \text{ and } \quad 
L_{{\algBabr}} = \sum_{j = 1}^{J} c (Z_{j-1}, Z_j) \leq J.
\end{equation*}
Thus, we have
\begin{equation}
L_{{\algBabr}} - L_{best} \leq J.
\label{eq: switching cost bound}
\end{equation}

Now it suffices to only consider the difference between rewards.
An important observation is that we group the timeslots into batches,  calculate the strategy probabilities only at the beginning of each batch, and use the average value of entire batch to update weight. Therefore, each batch can be considered as a round in conventional MAB. With this consideration, we introduce the notations below for our proofs,

\begin{equation}
\bar{\mathsf{g}}_{s, j} 
= \sum_{k \in s}\bar{\mathsf{f}}_{k, j} , \quad 
	\bar{\mathsf{G}}_{{\algBabr}} \myeq \sum_{j = 1}^{J}\bar{\mathsf{g}}_{Z_j, j}  \quad \text{ and } 
	\quad 
	\bar{\mathsf{G}}_{best} \myeq \max_{s \in \mathcal{S}}\sum_{j = 1}^{J}\bar{\mathsf{g}}_{s, j}.
	\label{eq: cumu str reward}
\end{equation}

Note that
\begin{align*}
\bar{\mathsf{g}}_{s, j} 
&= \sum_{k \in s}\bar{\mathsf{f}}_{k, j} 
= \frac{1}{\tau}\sum_{k \in s}\mathsf{f}_{k, j} 
= \frac{1}{\tau} \sum_{k \in s} \sum_{i = 1}^{\tau}f_{k,[j] +r} 
\\
&= \frac{1}{\tau}  \sum_{i = 1}^{\tau} \sum_{k \in s} f_{k,[j] +r}
= \frac{1}{\tau} \sum_{i = 1}^{\tau} g_{s, [j] +r}.
\end{align*}
Thus we have
\begin{equation}
G_{{\algBabr}} = \sum_{t = 1}^{T}g_{X_t, t} = \tau \sum_{j = 1}^{J}\bar{\mathsf{g}}_{s, j} = \tau \bar{\mathsf{G}}_{{\algBabr}}.
\label{eq: algo batch}
\end{equation}
Similarly,
\begin{equation}	
G_{best} = \tau \bar{\mathsf{G}}_{best}.
\label{eq: best batch}
\end{equation}

%
%
We now provide the bound of the expected difference between $\bar{\mathsf{G}}_{{\algBabr}}$ and $\bar{\mathsf{G}}_{best}$. 
\begin{lemma}
	\label{lemma3}
	For any type of adversaries, any $T > 0$, and any $\eta > 0$, we have
	\begin{equation}
	\Expect \left[
		\bar{\mathsf{G}}_{best} 
		-
	\bar{\mathsf{G}}_{{\algBabr}}
	\right]
	\leq
	\frac{\eta J S}{2}  + \frac{ \ln S}{\eta},
	\label{lemma3eq}
	\end{equation}
	where $\ln$ is the natural logarithm function.
\end{lemma}


\begin{proof}
	The proof of this lemma is based on \cite[Theorem 3.1]{bubeck2012regret} with necessary modifications and extensions.
	
	First we analyze $\frac{\mathsf{W}_{j+1}}{\mathsf{W}_{j}}$.
	For any sequences $Z_1, \ldots, Z_j$ generated by {\algB}, we have
	\begin{align*} 
		\frac{\mathsf{W}_{j+1}}{\mathsf{W}_{j}} 
		=~& \sum_{s \in \mathcal{S}}
		\frac{\mathsf{w}_{s, j+1}}{\mathsf{W}_j}
		\\
		=~& \sum_{s \in \mathcal{S}}
		\frac{\mathsf{w}_{s, j}}{\mathsf{W}_j}
		\exp\left(- \eta \bar{\mathsf{g}}^{\prime}_{s,j} \right)
		\\
		=~& \sum_{s \in \mathcal{S}}
		\mathsf{p}_{s, j} 
		\exp\left(-\eta \bar{\mathsf{g}}^{\prime}_{s,j}\right)
		\stepcounter{equation}\tag{\theequation}\label{eq: w to p loss}
		\\
		\leq~& \sum_{s \in \mathcal{S}}
		\mathsf{p}_{s, j} 
		\left(1 -\eta\bar{\mathsf{g}}^{\prime}_{s,j} 
		+ \frac{1}{2} \left(\eta\bar{\mathsf{g}}^{\prime}_{s,j}\right)^2\right)
		\stepcounter{equation}\tag{\theequation}\label{eq: e-x}
		\\
		\leq~& 1
		- \eta \sum_{s \in \mathcal{S}} 
		\left(  \mathsf{p}_{s, j} \bar{\mathsf{g}}^{\prime}_{s,j} \right)
		+ \frac{\eta ^ 2}{2} 
		\sum_{s \in \mathcal{S}}
		\mathsf{p}_{s, j}
		\left(\bar{\mathsf{g}}^{\prime}_{s,j}\right)^2 ,
		\stepcounter{equation}\tag{\theequation}\label{eq: delete p loss}
	\end{align*}
	where \eqref{eq: w to p loss} uses the definition of $\mathsf{p}_{s, j}$, and \eqref{eq: e-x} holds by the fact that for $x \geq 0$, $e^{-x} \leq 1 - x + \frac{1}{2}x^2$.

	Next, we bound~\eqref{eq: delete p loss} by bounding $	\sum_{s \in \mathcal{S}}
	\mathsf{p}_{s, j}\bar{\mathsf{g}}^{\prime}_{s,j}$ and $\sum_{s \in \mathcal{S}}
	\mathsf{p}_{s, j}
	\left(\bar{\mathsf{g}}^{\prime}_{s,j}\right)^2$.
	\begin{align*}\nonumber
	\sum_{s \in \mathcal{S}}
	\mathsf{p}_{s, j}\bar{\mathsf{g}}^{\prime}_{s,j}
	& = \sum_{s \in \mathcal{S}}
	\left(
	\mathsf{p}_{s, j}
	\sum_{k \in s}
	\bar{\mathsf{f}}^{\prime}_{k,j}
	\right)
	\\
	& = \sum_{k \in \mathcal{K}}
	\left(
	\bar{\mathsf{f}}^{\prime}_{k,j}
	\sum_{s: k \in s}
	\mathsf{p}_{s, j}
	\right)
	\\
	& = \sum_{k \in \mathcal{K}}
	\left(
	\bar{\mathsf{f}}^{\prime}_{k,j}
	\mathsf{q}_{k, j}
	\right)
	\\
	& = \sum_{k \in \mathcal{K}}
	\left( 1/l - \bar{\mathsf{f}}_{k,j} \right)
	\\
	& = \sum_{k \in Z_j}
	\left( 1/l - \bar{\mathsf{f}}_{k,j} \right) 
	+ \sum_{k \in \mathcal{K}\backslash Z_j}
	\left( 1/l - \bar{\mathsf{f}}_{k,j} \right)
	\\
	& \geq 1 - \bar{\mathsf{g}}_{Z_j,j},
	\stepcounter{equation}\tag{\theequation}\label{eq: delete K-Zj}
	\end{align*}
	
%
			\begin{align*}
			\sum_{s \in \mathcal{S}}
			\mathsf{p}_{s, j}\left(\bar{\mathsf{g}}^{\prime}_{s,j}\right)^2
			& = \sum_{s \in \mathcal{S}}
			\left(
			\mathsf{p}_{s, j}
			\left(
			\sum_{k \in s}
			\bar{\mathsf{f}}^{\prime}_{k,j}
			\right)^2
			\right)
			\\
			& \leq \sum_{s \in \mathcal{S}}
			\left(
			\mathsf{p}_{s, j}
			\cdot
			l
			\cdot
			\sum_{k \in s}
			\left(
			\bar{\mathsf{f}}^{\prime}_{k,j}
			\right)^2
			\right)
			\stepcounter{equation}\tag{\theequation}\label{eq: mean inequality}
			\\
			& = l \cdot
			\sum_{k \in \mathcal{K}}
			\left(
			\left(
			\bar{\mathsf{f}}^{\prime}_{k,j}
			\right)^2
			\sum_{s: k \in s}
			\mathsf{p}_{s, j}
			\right)
			\\
			& = l \cdot
			\sum_{k \in \mathcal{K}}
			\left(
			\left(
			\bar{\mathsf{f}}^{\prime}_{k,j}
			\right)^2
			\mathsf{q}_{k, j}
			\right)
			\\
			& = l \cdot
			\sum_{k \in \mathcal{K}}
			\left(
			\bar{\mathsf{f}}^{\prime}_{k,j}
			\cdot
			\frac{\left( 1/l - \bar{\mathsf{f}}_{k,j} \right)}{\mathsf{q}_{k, j}}
			\cdot
			\mathsf{q}_{k, j}
			\right)
			\\
			& \leq l \cdot
			\sum_{k \in \mathcal{K}}
			\left(
			\bar{\mathsf{f}}^{\prime}_{k,j}
			\cdot
			\frac{1}{f}
			\right)
			\\
			& = 
			\sum_{k \in \mathcal{K}}
			\bar{\mathsf{f}}^{\prime}_{k,j}
			\\
			& = \sum_{k \in Z_j}
			\bar{\mathsf{f}}^{\prime}_{k,j} 
			+ \sum_{k \in \mathcal{K}\backslash Z_j}
			\bar{\mathsf{f}}^{\prime}_{k,j}
			\\
			& = \bar{\mathsf{g}}^{\prime}_{Z_j,j}
			,
			\end{align*}
		where \eqref{eq: mean inequality} holds as a special case of the Cauchy-Schwarz inequality,
		

	Thus, we have
	\begin{align*}\nonumber
		\sum_{s \in \mathcal{S}}
		\mathsf{p}_{s, j}\bar{\mathsf{g}}^{\prime}_{s,j}
		\geq \bar{\mathsf{g}}_{Z_j,j}
		\text{ and }
		\sum_{s \in \mathcal{S}}
		\mathsf{p}_{s, j}\left(\bar{\mathsf{g}}^{\prime}_{s,j}\right)^2
		\leq
		\bar{\mathsf{g}}^{\prime}_{Z_j,j}
		\stepcounter{equation}\tag{\theequation}\label{eq: delete K-Zj again loss}.
	\end{align*}
	
	Combining~\eqref{eq: delete p loss}
	and~\eqref{eq: delete K-Zj again loss}, we have
	\begin{equation}
	\frac{\mathsf{W}_{j+1}}{\mathsf{W}_{j}} 
	\leq 
	1 
	- \eta
	(1-\bar{\mathsf{g}}_{Z_j,j} )
	+ \frac{\eta^2}{2}
	\bar{\mathsf{g}}^{\prime}_{Z_j,j}.
	\end{equation}
	
	Taking the log of both sides and using $1+x \leq e^x$ gives
	\begin{equation}
	\ln \frac{\mathsf{W}_{j+1}}{\mathsf{W}_{j}} 
	\leq
	- \eta
	(1-\bar{\mathsf{g}}_{Z_j,j} )
	+ \frac{\eta^2}{2}
	\bar{\mathsf{g}}^{\prime}_{Z_j,j}.
	\end{equation}
	
	Summing over $j$, we then get 
	\begin{equation}
	\label{eq: upper loss}
	\ln \frac{\mathsf{W}_{J+1}}{\mathsf{W}_{1}} 
	\leq
	- \eta
	(J - \bar{\mathsf{G}}_{{\algBabr}})
	+ \frac{\eta^2}{2}
	\sum_{j = 1}^{J}
	\bar{\mathsf{g}}^{\prime}_{Z_j,j}.
	\end{equation}
	
	Now we consider the lower bound of $\ln \frac{\mathsf{W}_{J+1}}{\mathsf{W}_{1}}$.
	For any strategy $s$,
	\begin{align*}
		\ln \frac{\mathsf{W}_{J+1}}{\mathsf{W}_{1}} 
		& \geq
		\ln \frac{\mathsf{w}_{s,j+1}}{\mathsf{W}_{1}} \\
		& =
		\ln \frac{
			\mathsf{w}_{s, 1}
			\exp
			\left(
			-\eta 
			\sum_{j=1}^{J}\bar{\mathsf{g}}^{\prime}_{s, j}
			\right)
		}
		{I\mathsf{w}_{s, 1}} \\
		& = 
		-\eta \sum_{j=1}^{J}\bar{\mathsf{g}}^{\prime}_{s, j} - \ln S. 
	\end{align*}
	Since the above inequality holds for any strategy $s$, we get
	\begin{align}
		\ln \frac{\mathsf{W}_{J+1}}{\mathsf{W}_{1}} 
		& \geq 
		-\eta
		\min_{s \in \mathcal{S}}
		\sum_{j=1}^{J}\bar{\mathsf{g}}^{\prime}_{s, j} - \ln S.
		\label{eq: lower}
	\end{align}
	
	Combining \eqref{eq: upper loss} and \eqref{eq: lower}, we have
	\begin{equation*}
		-\eta
		(J-\bar{\mathsf{G}}_{{\algBabr}})
		+ \frac{\eta^2}{2}
		\sum_{j = 1}^{J}
		\bar{\mathsf{g}}^{\prime}_{Z_j,j}
		\geq 
		-\eta
		\min_{s \in \mathcal{S}}
		\sum_{j=1}^{J}\bar{\mathsf{g}}^{\prime}_{s, j} - \ln S. 
	\end{equation*}
	Dividing both sides by $-\eta$ and moving terms, we have
	\begin{equation}
	\label{eq: before exp loss}
	(J-\bar{\mathsf{G}}_{{\algBabr}})
	-
	\min_{s \in \mathcal{S}}
	\sum_{j=1}^{J}\bar{\mathsf{g}}^{\prime}_{s, j} 
	\leq
	\frac{\eta}{2}
	\sum_{j = 1}^{J}
	\bar{\mathsf{g}}^{\prime}_{Z_j,j}
	+
	\frac{\ln S}{\eta}.	 
	\end{equation}
	
	We take expectations of both sides in \eqref{eq: before exp loss} with respect to the distribution of $Z_1, Z_2, \ldots, Z_j$. For the conditional expected value of each $\bar{\mathsf{g}}^{\prime}_{s, j} $, we have
	
	\begin{align*}
		& 
		\Expect \left[
		\bar{\mathsf{g}}^{\prime}_{s, j} | Z_1, Z_2, \ldots, Z_{j-1} 
		\right]
		\\
		=~& \mathsf{p}_{s,j} \cdot \frac{1-\bar{\mathsf{g}}_{s, j}}{\mathsf{p}_{s,j}}  +
		(1 - \mathsf{p}_{s,j}) \cdot
		\Expect \left[
		\bar{\mathsf{g}}^{\prime}_{s, j} 
		|
		s\neq Z_j
		\right]
		\\
		=~& 
		\mathsf{p}_{s,j} \frac{1- \bar{\mathsf{g}}_{s, j}}{\mathsf{p}_{s,j}}
		\\
		=~&1- \bar{\mathsf{g}}_{s, j},
	\end{align*}
	where the first equation is due to the fact that $Z_j = s$ with probability $\mathsf{p}_{s,j}$,
	and the second equation is due to that $\bar{\mathsf{g}}^{\prime}_{s, j} 
	= 
	\sum_{k \in s}
	\bar{\mathsf{f}}^{\prime}_{k,j}
	= 0
	$
	when $s \neq Z_j$.
	
	
	Similarly,
	\begin{equation*}
		\Expect \left[
		\bar{\mathsf{g}}^{\prime}_{Z_j, j} | Z_1, Z_2, \ldots, Z_{j-1} 
		\right]
		= 1 - \bar{\mathsf{g}}_{Z_j, j}.
	\end{equation*}
	Therefore, we have
	\begin{align*}
		\Expect \left[
		\bar{\mathsf{G}}_{best}
		-
		\bar{\mathsf{G}}_{{\algBabr}}
		\right]
		\leq~& 
		\frac{\eta}{2}
		\sum_{j = 1}^{J}
		\bar{\mathsf{g}}_{Z_j,j}
		+
		\frac{\ln S}{\eta} 
		\\
		\leq~& 
		\frac{\eta}{2}
		J
		S
		+
		\frac{\ln S}{\eta} .
	\end{align*}
\end{proof}

Now taking  into consideration the bound of switching costs, \eqref{eq: switching cost bound}, we have the following theorem: 
\begin{theorem}
	\label{thm3}
	For any type of adversaries, the expected weak regret of {\algB} is $O(T^{\frac{2}{3}})$ with parameters
	$\eta = A_\eta T^{-\frac{1}{3}} > 0 \text{ and }   \tau =  A_\tau T^{\frac{1}{3}} \in \left[1, T\right],$
	where $A_\gamma$ and $A_\tau$ are constants. Specifically, 
	\begin{equation}
	\label{eq: bound Gbest loss}
	\Expect\left[R_{{\algBabr}}\right] 
	\leq
	\left(\frac{A_\eta S}{2} + \frac{A_\tau \ln S}{A_\eta} + \frac{1}{A_\tau}\right) T^{\frac{2}{3}}.
	\end{equation}
\end{theorem}
\begin{proof}
	\begin{align*}
		\nonumber
		\Expect\left[R_{{\algBabr}}\right]
		=~& \Expect\left[ 
		G_{best} - L_{best} 
		-G_{{\algBabr}} + L_{{\algBabr}}	
		\right]
		\\
		\leq~&
		\Expect\left[ G_{best} -  G_{{\algBabr}} \right ] + J
		\\
		=~&
		\Expect\left[
		\tau\bar{\mathsf{G}}_{best} 
		-
		\tau \bar{\mathsf{G}}_{{\algBabr}}
		\right] 
		+ J
		\\
		=~&
		\frac{\eta J	S \tau}{2}
		+
		\frac{\tau \ln S}{\eta} 
		+ J 
		\\
		=~&
		\frac{\eta 	S T}{2}
		+
		\frac{\tau \ln S}{\eta} 
		+ T/\tau 
		\\
		\nonumber
		=~&
		\left(
		\frac{A_\eta 	S }{2}
		+
		\frac{A_\tau \ln S}{A_\eta} 
		+ \frac{1}{A_\tau}
		\right) T^{\frac{2}{3}} . 
		\qedhere
	\end{align*}
\end{proof}

For better understanding of Theorem \ref{thm3}, we now give a specific bound by choosing particular parameters. 
\begin{corollary}
	\label{coro3}
	For any type of adversaries, when $T \geq \frac{1}{2} S \ln S$, with parameters
		$\eta = \sqrt[3]{\frac{4\ln S}{S^2T}}$
		and
		$\tau = \sqrt[3]{\frac{2T}{S\ln S}}$,
	the expected weak regret of {\algB} is
	\begin{equation}
	\Expect\left[R_{{\algBabr}}\right] 
	\leq
	3\left(\frac{1}{2} S\ln S\right)^{\frac{1}{3}} T^{\frac{2}{3}}.
	\end{equation}
\end{corollary}

\begin{proof}
	Substituting the parameters in \eqref{eq: bound Gbest loss}, we have the immediate result.
\end{proof}

\noindent \textit{Remark.} The expected weak regret of {\algB} is bounded by $O(T^{2/3})$. This upper bound matches the lower bound of the MAB-SC problem proved in~\cite{dekel2014bandits} which is used to model the {\problemName} in this paper. Thus, {\algB} is asymptotically optimal. If we calculate the normalized weak regret $\frac{R_{\algBabr}}{T}$, i.e., amortizing the regret to every timeslot, then it is clear that the expected value of normalized weak regret converges to 0 as $T$ approaches to $\infty$.


	\section{Spectrum Monitoring Algorithm with Bounded Weak Regret}

We have already proved that {\algB} is an effective online spectrum monitoring algorithms with expected normalized regret converging to 0. Though the \textit{expectation} provides a quite legitimate estimate on the performance of {\algA} and {\algB}, the actual value of weak regret may sometimes deviate a lot from the expected bound as expectation just represents the mean. In this section, we propose the improved algorithm, {\algC}, whose actual weak regrets are bounded by $O(T^{2/3})$ with any user-defined confidence level.
	
%
%
\subsection{{\algC}}\label{ssec: algC}

%
%




\subsubsection{Algorithm Design}

{\algC}  is designed similar to {\algB}. However, it takes four parameters $\tau$, $\gamma$, $\beta$ and $\eta$ as input, where $\gamma$ is used to calculate strategy probabilities, $\beta$ is used to calculate average channel scores, and $\eta$ is used to update strategy weights. 
%

\textbf{Calculating Strategy Probability $\mathsf{p}_{s, j}$.}
For calculating  strategy probabilities, we introduce a new concept called \textit{covering strategy set}. A covering strategy set $\mathcal{C} \subset \mathcal{S}$ is a set of strategy that \textit{cover}s all channels $\mathcal{K}$, where a channel $k \in \mathcal{K}$ is covered by $\mathcal{C}$ if there is a strategy $s \in \mathcal{C}$ such that $k \in s$.
In {\algC}, we randomly construct a minimal covering strategy set whose size $C \myeq |\mathcal{C}|$ is less than or equal to $K$. 
The probability of each strategy $s$ is calculated by 
\begin{equation}\label{eq: strategy prob in +}
\mathsf{p}_{s, j} = 
(1-\gamma)\frac{\mathsf{w}_{s, j}}{\mathsf{W}_j} + 
\frac{\gamma}{C}\mathds{1}_{s \in \mathcal{C}},
\end{equation}
where $\mathds{1}_{s \in \mathcal{C}}$ is an indicator function which has the value 1 if $s \in \mathcal{C}$; 0 otherwise.
In this way, the strategies in the covering set are more likely to be chosen than others.  As a result, {\algC} can explore all channels more quickly, and thus reveal the best channels sooner, which expedites the exploration for the best strategy. 

\textbf{Choosing Strategy $Z_j$ and Monitoring Spectrum.}
This part is the same as {\algB}.

\textbf{Updating Strategy Weight $w_{s, j+1}$.}
For calculating average channel scores, 
we introduce a new parameter $\beta$ and have
$
	\bar{\mathsf{f}}^{\prime}_{k, j} = 
	\frac{\bar{\mathsf{f}}_{k, j} + \beta}{{\mathsf{q}_{k,j}}},
	$
where ${\mathsf{q}_{k,j}}$ is the channel probability of $k \in \mathcal{K}$ in batch $j$.
By \eqref{eq: strategy prob in +}, we have
 is 
$
{\mathsf{q}_{k,j}}
= {\sum_{s: k \in s}\mathsf{p}_{s,j}}
= (1-\gamma)\frac{\sum_{s: k \in s}\mathsf{w}_{s, j}}
{\mathsf{W}_j} + \frac{\gamma C_k}{C},
$
where $C_k$ is the number of strategies in the covering strategy set and containing channel k, i.e., $C_k = |\{s | s \in \mathcal{C} \land k \in s\}|$. 

The average channel score use $\mathsf{q}_{k,j}$ as the denominator in order to  compensate the rewards of channels with low probabilities. Among the channels receiving rewards, those with lower probabilities can obtain higher average channel scores, and therefore higher channel weights.
Note that we could also gain rewards on an unmonitored channel if we had monitored it,
which indicates that the average channel score of that channel should be positive. 
With this concern, we use parameter $\beta$ to reduce the bias between monitored and unmonitored channels.


Then the channel weight of $k$ is
\begin{equation}
\label{eq: channel weight in +}
\mathsf{h}_{k, j+1} = \mathsf{h}_{k, j} 
\exp 
\left(
\eta \bar{\mathsf{f}}^{\prime}_{k, j-1}
\right),
\end{equation}
where $\mathsf{h}_{k, 1} = 1$ for all $k \in \mathcal{K}$.
Thus, the strategy weight of $s$ is updated by
\begin{equation}
\label{eq： weight update in +}			
\mathsf{w}_{s, j+1} = \mathsf{w}_{s, j} \exp (\eta \bar{\mathsf{g}}^{\prime}_{s, j}),
\end{equation}
where $\bar{\mathsf{g}}^{\prime}_{s, j} $ is  the average strategy score of $s$ in batch $j$ and can be calculated directly by
\begin{equation}\label{eq: average strategy score in +}
\bar{\mathsf{g}}^{\prime}_{s, j} 
= \sum_{k \in s} 
\frac{
	\frac{1}{\tau} 
	\sum_{i = 1}^{\tau} f_{k, [j] + i} + \beta}{(1-\gamma)\frac{\sum_{s: k \in s}\mathsf{w}_{s, j}}
	{\mathsf{W}_j} + \frac{\gamma C_k}{C}}.
\end{equation}

%
%

\subsubsection{Performance Analysis} 

Since {\algC} only update the monitoring strategy across batches, each batch can be regarded as a round in conventional MAB. We define $\bar{\mathsf{G}}_{\algCabr}  \myeq \sum_{j = 1}^{J}\bar{\mathsf{g}}_{Z_j, j}$, where $Z_j$ is {\algC}'s chosen strategy for each batch. Then we have the following lemma.

	We first  introduce the following two notations,
	\begin{align}
	\bar{\mathsf{F}}_{k, n} \myeq \sum _{j =1}^{n} \bar{\mathsf{f}}_{k, j} 
	\quad &\mbox{and} \quad
	\bar{\mathsf{F}}^{\prime}_{k, n} \myeq \sum _{j =1}^{n} \bar{\mathsf{f}}^{\prime}_{k, j} 
	\quad \mbox{for} \quad k \in \mathcal{K},
	\end{align}
	where $n$ is an arbitrary batch, $\bar{\mathsf{f}}_{k, j}= \frac{1}{\tau}\sum_{i = 1}^{\tau} f_{k, [j] + i}$ and $\bar{\mathsf{f}}^{\prime}_{k, j} = 
	\frac{\bar{\mathsf{f}}_{k, j} + \beta}{{\mathsf{q}_{k,j}}}$.
	
	We then prove the following lemma.
	\begin{lemma}
		
	For any type of adversaries, for any $\delta \in (0, 1)$, $\beta \in (0, 1)$ and $k \in \mathcal{K}$ in {\algC}, we have
	\begin{equation}
	\Probability \left[ \bar{\mathsf{F}}_{k, n} \geq \bar{\mathsf{F}}^{\prime}_{k, n} + \frac{1}{\beta} \ln \frac{K}{\delta}\right] 
	\leq
	\frac{\delta}{K}.
	\label{reward chernoff}
	\end{equation}
	\end{lemma}
%
	 
	 \begin{proof}
	 	We prove \eqref{reward chernoff} based on~\cite[Lemma 2]{gyorgy2007line}.
	 	\quad
	 	Note that $\bar{\mathsf{f}}^{\prime}_{k, j}$ for different batches are generated independently and $\bar{\mathsf{F}}_{k, n}$ is the sum of these independent random variables. By the Chernoff bound, we have
	 	\begin{equation*}
	 	\Probability \left[\bar{\mathsf{F}}_{k, n} \geq \bar{\mathsf{F}}^{\prime}_{k, n} + u\right]
	 	\leq
	 	\exp(-uv) \Expect \left[ \exp \left(v \left(\bar{\mathsf{F}}_{k, n} - \bar{\mathsf{F}}^{\prime}_{k, n}\right)\right) \right],
	 	\end{equation*}
	 	for any $k \in \mathcal{K}$, any $u > 0$, and any $v > 0$.
	 	Let $u = \frac{1}{\beta} \ln \frac{K}{\delta}$ and $v = \beta$, then we have
	 	\begin{align*}
	 	&\exp(-uv) \Expect \left[ \exp \left(v \left(\bar{\mathsf{F}}_{k, n} - \bar{\mathsf{F}}^{\prime}_{k, n}\right)\right) \right]
	 	\\
	 	=~&
	 	\exp(-\ln \frac{K}{\delta}) \Expect \left[ \exp \left(\beta \left(\bar{\mathsf{F}}_{k, n} - \bar{\mathsf{F}}^{\prime}_{k, n}\right)\right) \right]
	 	\\
	 	=~&
	 	\frac{\delta}{K} \Expect \left[ \exp \left(\beta \left(\bar{\mathsf{F}}_{k, n} - \bar{\mathsf{F}}^{\prime}_{k, n}\right)\right) \right].
	 	\end{align*}
	 	Thus, to prove \eqref{reward chernoff},  it suffices to prove that for all $n$,
	 	\begin{equation*}
	 	\Expect \left[ \exp \left(\beta \left(\bar{\mathsf{F}}_{k, n} - \bar{\mathsf{F}}^{\prime}_{k, n}\right)\right) \right]
	 	\leq 1.
	 	\end{equation*}	
	 	Define
	 	\begin{equation*}
	 	D_n \myeq \exp \left(\beta \left(\bar{\mathsf{F}}_{k, n} - \bar{\mathsf{F}}^{\prime}_{k, n}\right)\right).
	 	\end{equation*}
	 	We first show that $\Expect_n \left[D_n\right] \leq D_{n-1}$ for $n \geq 2$, where $\Expect_n$ denotes the conditional expectation $\Expect \left[\cdot | Z_1, Z_2, \ldots Z_{n-1}\right]$. Note that
	 	\begin{equation*}
	 	D_n = D_{n-1} \exp \left(\beta \left(\bar{\mathsf{f}}_{k, n} - \bar{\mathsf{f}}^{\prime}_{k, n}\right)\right).
	 	\end{equation*}
	 	Taking conditional expectations, we obtain
	 	\begin{align*}
	 	&\Expect_n \left[D_n\right] \nonumber\\
	 	=~&
	 	D_{n-1} \Expect_n \left[\exp \left(\beta \left(\bar{\mathsf{f}}_{k, n} - \bar{\mathsf{f}}^{\prime}_{k, n}\right)\right)\right]
	 	\\
	 	=~&	
	 	D_{n-1} \Expect_n \left[\exp \left(\beta \left(\bar{\mathsf{f}}_{k, n} - \frac{\mathds{1}_{k \in Z_n}\bar{\mathsf{f}}_{k, n} + \beta}{\mathsf{q}_{k,n}}\right)\right)\right]
	 	\\
	 	=~&
	 	D_{n-1} \exp \left(-\frac{\beta^2}{\mathsf{q}_{k, n}}\right)
	 	\Expect_n \left[\exp \left(\beta \left(\bar{\mathsf{f}}_{k, n} - \frac{\mathds{1}_{k \in Z_n}\bar{\mathsf{f}}_{k, n}}{\mathsf{q}_{k,n}}\right)\right)\right]
	 	\\
	 	\leq~&
	 	D_{n-1} \exp \left(-\frac{\beta^2}{\mathsf{q}_{k, n}}\right) \Expect_n \left[1 
	 	+ \beta \left(\bar{\mathsf{f}}_{k, n} - \frac{\mathds{1}_{k \in Z_n}\bar{\mathsf{f}}_{k, n}}{\mathsf{q}_{k,n}}\right)		
	 	\right. \nonumber
	 	\\
	 	& \left. 		
	 	+ \beta^2 \left(\bar{\mathsf{f}}_{k, n} - \frac{\mathds{1}_{k \in Z_n}\bar{\mathsf{f}}_{k, n}}{\mathsf{q}_{k,n}}\right)^2 \right]
	 	\stepcounter{equation}\tag{\theequation}\label{eq: exp in thm3}
	 	\\
	 	=~&
	 	D_{n-1} \exp \left(-\frac{\beta^2}{\mathsf{q}_{k, n}}\right) \Expect_n \left[1 + \beta^2 \left(\bar{\mathsf{f}}_{k, n} - \frac{\mathds{1}_{k \in Z_n}\bar{\mathsf{f}}_{k, n}}{\mathsf{q}_{k,n}}\right)^2 \right]
	 	\stepcounter{equation}\tag{\theequation}\label{eq: expect in thm3}
	 	\\
	 	\leq~&
	 	D_{n-1} \exp \left(-\frac{\beta^2}{\mathsf{q}_{k, n}}\right) \Expect_n \left[1 + \beta^2 \left(\frac{\mathds{1}_{k \in Z_n}\bar{\mathsf{f}}_{k, n}}{\mathsf{q}_{k,n}}\right)^2 \right]
	 	\\
	 	\leq~&
	 	D_{n-1} \exp \left(-\frac{\beta^2}{\mathsf{q}_{k, n}}\right) \left(1 +  \frac{\beta^2}{\mathsf{q}_{k,n}} \right)
	 	\\
	 	\leq~&
	 	D_{n-1},
	 	\stepcounter{equation}\tag{\theequation}\label{eq: final in thm3}
	 	\end{align*}
	 	where $\mathds{1}_{k \in Z_n} = 1$ if $k \in Z_n$ and 0 otherwise; \eqref{eq: exp in thm3} holds because $\beta < 1$, $\bar{\mathsf{f}}_{k, n} - \frac{\mathds{1}_{k \in Z_n}\bar{\mathsf{f}}_{k, n}}{\mathsf{q}_{k,n}} \leq 1$ and $e^x \leq 1 + x + x^2$ for $x \leq 1$; \eqref{eq: expect in thm3} follows from $\Expect_t \left[\frac{\mathds{1}_{k \in Z_n}\bar{\mathsf{f}}_{k, n}}{\mathsf{q}_{k,n}}\right] = \bar{\mathsf{f}}_{k, n}$; \eqref{eq: final in thm3} holds by the inequality $1 +x \leq e^x$.
	 	Taking expectations on both sides proves
	 	\begin{equation*}
	 	\Expect \left[D_n\right] \leq \Expect \left[D_{n-1}\right].
	 	\end{equation*}
	 	A similar approach shows that $\Expect \left[D_1\right] \leq 1$, which implies $\Expect \left[D_n\right] = \Expect \left[ \exp \left(\beta \left(\bar{\mathsf{F}}_{k, n} - \bar{\mathsf{F}}^{\prime}_{k, n}\right)\right) \right] \leq 1$ as desired.
	 \end{proof}
	 
	
%

\begin{lemma}\label{SW+ reward}
	For any type of adversaries, any $T > 0$, $\gamma \in (0, 1/2)$, $\tau \in \left[1, T\right]$, $ \beta \in (0, 1)$, and $\eta > 0$ satisfying $2\eta lC \leq \gamma$, we have 
	\begin{equation*}
	\bar{\mathsf{G}}_{best} - \bar{\mathsf{G}}_{\algCabr}
	\leq
	\gamma J + 2\eta lCJ + \frac{l}{\beta}\ln\frac{K}{\delta} + \frac{\ln S}{\eta} + \beta KJ.
	\end{equation*}	
	with probability at least $1-\delta$ for any $\delta \in (0,1)$.
\end{lemma}

\begin{proof}
	First, we show $0 \leq \eta\bar{\mathsf{g}}_{s, j} \leq 1$. It is easy to notice that $\eta \bar{\mathsf{g}}^{\prime}_{s, j} \geq 0$.
	In addition, we have 
	\begin{align*}
	\eta \bar{\mathsf{g}}^{\prime}_{s, j}
	&=
	\eta \sum_{k \in s}\bar{\mathsf{f}}^{\prime}_{k, j}
	\leq
	\eta \sum_{k \in s}\frac{\bar{\mathsf{f}}_{k, j} + \beta}{\bar{\mathsf{q}}_{s, j}} 
	\\&\leq
	\eta \sum_{k \in s}\frac{1 + \beta}{\frac{\gamma}{C}} 
	\leq
	\frac{(1 + \beta)\eta lC}{\gamma}
	\leq
	1 , 
	\end{align*}
	where the second inequality is due to $\mathsf{q}_{k,j} \geq \frac{\gamma}{C}$  for all $k \in \mathcal{K}$, and the last inequality is due to $2\eta lC \leq \gamma$. 
	
	Then we analyze $\frac{\mathsf{W}_{j+1}}{\mathsf{W}_{j}}$.
	For any sequence $Z_1, \ldots, Z_j$ generated by {\algC}, we have
	\begin{align*}
	&\frac{\mathsf{W}_{j+1}}{\mathsf{W}_{j}} 
	\\
	=~& \sum_{s \in \mathcal{S}}
	\frac{\mathsf{w}_{s, j+1}}{\mathsf{W}_j}
	\\
	=~& \sum_{s \in \mathcal{S}}
	\frac{\mathsf{w}_{s, j}}{\mathsf{W}_j}
	\exp\left(\eta\bar{\mathsf{g}}^{\prime}_{s,j} \right)
	\\
	=~& \sum_{s \in \mathcal{S}}
	\frac{\mathsf{p}_{s, j} - \frac{\gamma}{C} \mathds{1}_{s \in \mathcal{C}}}{1 - \gamma} \exp\left(\eta \bar{\mathsf{g}}^{\prime}_{s,j}\right)
	\stepcounter{equation}\tag{\theequation}\label{eq: w to p in +}
	\\
	\leq~& \sum_{s \in \mathcal{S}}
	\frac{\mathsf{p}_{s, j} - \frac{\gamma}{C} \mathds{1}_{s \in \mathcal{C}}}{1 - \gamma}
	\left(1+ \eta \bar{\mathsf{g}}^{\prime}_{s,j} 
	+ \eta ^2 \left(\bar{\mathsf{g}}^{\prime}_{s,j}\right)^2\right)
	\stepcounter{equation}\tag{\theequation}\label{eq: e to 1 + x + x2 in +}
	\\
	\leq~& \frac{1 - \gamma}{1 - \gamma} 
	+ \sum_{s \in \mathcal{S}} \frac{\mathsf{p}_{s, j}}{1 - \gamma}
	\left(\eta \bar{\mathsf{g}}^{\prime}_{s,j} 
	+ \eta ^2 \left(\bar{\mathsf{g}}^{\prime}_{s,j}\right)^2\right)
	\\
	\leq~& 1 
	+ \frac{\eta}{1 - \gamma}
	\sum_{s \in \mathcal{S}}
	\mathsf{p}_{s, j}\bar{\mathsf{g}}^{\prime}_{s,j}
	+ \frac{\eta^2}{1 - \gamma}
	\sum_{s \in \mathcal{S}}
	\mathsf{p}_{s, j}\left(\bar{\mathsf{g}}^{\prime}_{s,j}\right)^2
	\stepcounter{equation}\tag{\theequation}\label{eq: delete p in +},
	\end{align*}
	where \eqref{eq: w to p in +} uses the definition of $\mathsf{p}_{s, j}$ in \eqref{eq: strategy prob in +}, and \eqref{eq: e to 1 + x + x2 in +} holds by the fact that $e^x \leq 1 + x + x^2$ for $0 \leq x \leq 1$. 
	
	Next we bound \eqref{eq: delete p in +}. For the second term, we have
	\begin{align*}
	\sum_{s \in \mathcal{S}}
	\mathsf{p}_{s, j}\bar{\mathsf{g}}^{\prime}_{s,j}
	& = \sum_{s \in \mathcal{S}}
	\left(
	\mathsf{p}_{s, j}
	\sum_{k \in s}
	\bar{\mathsf{f}}^{\prime}_{k,j}
	\right)
	\\
	& = \sum_{k \in \mathcal{K}}
	\left(
	\bar{\mathsf{f}}^{\prime}_{k,j}
	\sum_{s: k \in s}
	\mathsf{p}_{s, j}
	\right)
	\\
	& = \sum_{k \in \mathcal{K}}
	\left(
	\bar{\mathsf{f}}^{\prime}_{k,j}
	\mathsf{q}_{k, j}
	\right)
	\\
	& = \sum_{k \in \mathcal{K}}
	\left(\bar{\mathsf{f}}_{k,j} + \beta\right)
	\\
	& = \sum_{k \in Z_j}
	\left(\bar{\mathsf{f}}_{k,j} + \beta\right)
	+ \sum_{k \in \mathcal{K}\backslash Z_j}
	\left(\bar{\mathsf{f}}_{k,j} + \beta\right)
	\\
	& = \bar{\mathsf{g}}_{Z_j,j} + K \beta ,
	\stepcounter{equation}\tag{\theequation}\label{eq: delete K-Zj in +}
	\end{align*}
	where \eqref{eq: delete K-Zj in +} uses the definition of  average strategy reward
	 and the fact that $\bar{\mathsf{f}}_{k,j}$ is 0 when $k \notin Z_j$.
	For the second sum,
	\begin{align*}
	\sum_{s \in \mathcal{S}}
	\mathsf{p}_{s, j}\left(\bar{\mathsf{g}}^{\prime}_{s,j}\right)^2
	& = \sum_{s \in \mathcal{S}}
	\left(
	\mathsf{p}_{s, j}
	\left(
	\sum_{k \in s}
	\bar{\mathsf{f}}^{\prime}_{k,j}
	\right)^2
	\right)
	\\
	& \leq \sum_{s \in \mathcal{S}}
	\left(
	\mathsf{p}_{s, j}
	\cdot
	l
	\cdot
	\sum_{k \in s}
	\left(
	\bar{\mathsf{f}}^{\prime}_{k,j}
	\right)^2
	\right)
	\stepcounter{equation}\tag{\theequation}\label{eq: mean inequality in +}
	\\
	& = l \cdot
	\sum_{k \in \mathcal{K}}
	\left(
	\left(
	\bar{\mathsf{f}}^{\prime}_{k,j}
	\right)^2
	\sum_{s: k \in s}
	\mathsf{p}_{s, j}
	\right)
	\\
	& = l \cdot
	\sum_{k \in \mathcal{K}}
	\left(
	\left(
	\bar{\mathsf{f}}^{\prime}_{k,j}
	\right)^2
	\mathsf{q}_{k, j}
	\right)
	\\
	& = l \cdot
	\sum_{k \in \mathcal{K}}
	\left(
	\bar{\mathsf{f}}^{\prime}_{k,j}
	\cdot
	\frac{\mathds{1}_{k \in Z_j} \bar{\mathsf{f}}_{k,j} + \beta}{\mathsf{q}_{k, j}}
	\cdot
	\mathsf{q}_{k, j}
	\right)
	\\
	& \leq l \cdot (1 + \beta) \cdot
	\sum_{k \in \mathcal{K}}
	\bar{\mathsf{f}}^{\prime}_{k,j}
	\\
	& \leq l \cdot (1 + \beta) \cdot
	\sum_{s \in \mathcal{C}}
	\bar{\mathsf{g}}^{\prime}_{s,j},
	\stepcounter{equation}\tag{\theequation}\label{eq: delete K-Zj again in +}
	\end{align*}
	where \eqref{eq: mean inequality in +} holds as a special case of the Cauchy-Schwarz Inequality $\left(\sum_{i = 1}^{n}a_i \cdot 1\right)^2 \leq \left(\sum_{i = 1}^{n}a_i^2 \right) \left(\sum_{i = 1}^{n} 1^2\right)$, and \eqref{eq: delete K-Zj again in +} holds because covering strategy set $\mathcal{C}$ covers each channel at least once.

	Therefore, combining~\eqref{eq: delete p in +}, ~\eqref{eq: delete K-Zj in +}, and~\eqref{eq: delete K-Zj again in +}, we have
	\begin{equation*}
	\frac{\mathsf{W}_{j+1}}{\mathsf{W}_{j}} 
	\leq 
	1 
	+ \frac{\eta}{1 - \gamma}
	\left(\bar{\mathsf{g}}_{Z_j,j} + K \beta\right)
	+ \frac{\eta^2 l (1 + \beta)}{1 - \gamma}
	\sum_{s \in \mathcal{C}}
	\bar{\mathsf{g}}^{\prime}_{s,j} .
	\end{equation*}
	
	Taking the log of both sides and using $1+x \leq e^x$ gives
	\begin{equation*}
	\ln \frac{\mathsf{W}_{j+1}}{\mathsf{W}_{j}} 
	\leq
	\frac{\eta}{1 - \gamma}
	\left(\bar{\mathsf{g}}_{Z_j,j} + K \beta\right)
	+ \frac{\eta^2 l (1 + \beta)}{1 - \gamma}
	\sum_{s \in \mathcal{C}}
	\bar{\mathsf{g}}^{\prime}_{s,j} .
	\end{equation*}
	
	Summing over $j$ we then get 
	\begin{equation*}
	\ln \frac{\mathsf{W}_{J+1}}{\mathsf{W}_{1}} 
	\leq
	\frac{\eta}{1 - \gamma}
	\left(\bar{\mathsf{G}}_{\algCabr} + JK\beta \right)	
	+  \frac{\eta^2 l (1 + \beta)}{1 - \gamma}
	\sum_{j = 1}^{J}
	\sum_{s \in \mathcal{C}}
	\bar{\mathsf{g}}^{\prime}_{s,j} .
	\end{equation*}
	
	Note that 
	\begin{equation*}
	\sum_{j = 1}^{J}
	\sum_{s \in \mathcal{C}}
	\bar{\mathsf{g}}^{\prime}_{s,j}
	\leq
	C \max_{s \in \mathcal{S}} \sum _{j =1}^{J} \bar{\mathsf{g}}^{\prime}_{s,j}
	\leq
	C \max_{s \in \mathcal{S}} \bar{\mathsf{G}}^{\prime}_{s,J}.
	\end{equation*}	
	
	We have
	\begin{align}
	\ln \frac{\mathsf{W}_{J+1}}{\mathsf{W}_{1}} 
	\le
	\frac{\eta}{1 - \gamma}
	\left(\bar{\mathsf{G}}_{\algCabr} + JK\beta \right)	
	+  \frac{\eta^2 l (1 + \beta) C}{1 - \gamma}
	\max_{s \in \mathcal{S}} \bar{\mathsf{G}}^{\prime}_{s,J} . 
	\label{eq: upper in +}
	\end{align}
	
	Now we consider the lower bound of $\ln \frac{\mathsf{W}_{J+1}}{\mathsf{W}_{1}} $.
	For any strategy $s$,
	\begin{align*}
	\ln \frac{\mathsf{W}_{J+1}}{\mathsf{W}_{1}} 
	& \geq
	\ln \frac{\mathsf{w}_{s,j+1}}{\mathsf{W}_{1}}
	\\
	& =
	\ln \frac{
		\mathsf{w}_{s, 1}
		\exp
		\left(
		\eta
		\sum_{j=1}^{J}\bar{\mathsf{g}}^{\prime}_{s, j}
		\right)
	}
	{S\mathsf{w}_{s, 1}}
	\\
	& = 
	\eta \sum_{j=1}^{J}\bar{\mathsf{g}}^{\prime}_{s, j} - \ln S 
	\\
	& = 
	\eta \bar{\mathsf{G}}^{\prime}_{s, J} - \ln S .
	\end{align*}
	Since the above inequality holds for any strategy $s$, we get
	\begin{align}
	\ln \frac{\mathsf{W}_{J+1}}{\mathsf{W}_{1}} 
	& \geq 
	\eta
	\max_{s \in \mathcal{S}}
	\bar{\mathsf{G}}^{\prime}_{s, J} - \ln S .
	\label{eq: lower in +}
	\end{align}
	
	Combining \eqref{eq: upper in +} and \eqref{eq: lower in +}, we have
	\begin{equation}
	\bar{\mathsf{G}}_{\algCabr}
	\geq 
	(1-\gamma - \eta l (1 + \beta) C)
	\max_{s \in \mathcal{S}}
	\bar{\mathsf{G}}^{\prime}_{s, J}
	- \frac{1 - \gamma}{\eta}\ln S 
	- JK\beta .
	\end{equation}
	
	Note that
	\begin{equation*}
	\bar{\mathsf{G}}^{\prime}_{s, J} = \sum _{j =1}^{J} \bar{\mathsf{g}}^{\prime}_{s, j} = \sum _{j =1}^{J}  \sum_{k \in s} \bar{\mathsf{f}}^{\prime}_{k, j}= \sum_{k \in s} \bar{\mathsf{F}}^{\prime}_{k, n},
	\end{equation*}
	and that
	\begin{equation*}
	\bar{\mathsf{G}}_{s, J} = \sum _{j =1}^{J} \bar{\mathsf{g}}_{s, j} = \sum _{j =1}^{J}  \sum_{k \in s} \bar{\mathsf{f}}_{k, j}= \sum_{k \in s} \bar{\mathsf{F}}_{k, n}.
	\end{equation*}
	
	By using Lemma~\ref{reward chernoff} and applying Boole's inequality, we obtain that, with probability at least $1 - \delta$,
	\begin{align*}
	\bar{\mathsf{G}}_{\algCabr}
	&\geq 
	(1-\gamma - \eta l (1 + \beta) C)
	\left(
	\max_{s \in \mathcal{S}}
	\bar{\mathsf{G}}_{s, J} - \frac{l}{\beta} \ln \frac{K}{\delta}
	\right)	\nonumber
	\\
	& \quad - \frac{1 - \gamma}{\eta}\ln S 
	- JK\beta
	\\
	&\geq 
	(1-\gamma - \eta l (1 + \beta) C)
	\left(
	\bar{\mathsf{G}}_{best} - \frac{l}{\beta} \ln \frac{K}{\delta}
	\right)	\nonumber
	\\
	& \quad
	- \frac{1 - \gamma}{\eta}\ln S 
	- JK\beta,
	\end{align*}
	where $1-\gamma - \eta l (1 + \beta) C > 0$ because $\eta l (1 + \beta) C \leq 2 \eta l C \leq \gamma < 1/2$.
	
	Therefore,
	\begin{align*}
	&\bar{\mathsf{G}}_{best}
	-
	\bar{\mathsf{G}}_{\algCabr} \nonumber
	\\
	\leq~&
	(\gamma + \eta l (1 + \beta) C)
	\bar{\mathsf{G}}_{best}	
	\\
	& +
	\left(1-\gamma - \eta l (1 + \beta) C\right)
	\frac{l}{\beta} \ln \frac{K}{\delta} \nonumber
	+
	\frac{1 - \gamma}{\eta}\ln S 
	+ JK\beta
	\\
	\leq~& 
	\gamma J + 2\eta lC J
	+	
	\frac{l}{\beta} \ln \frac{K}{\delta}
	+\frac{\ln S }{\eta}
	+ \beta KJ,
	\end{align*}
	where the last inequity is due to the fact that $\bar{\mathsf{G}}_{best} \leq J$. 
\end{proof}

Next, we bound the difference between the cumulative strategy reward of the best fixed algorithm and that of {\algC}.
\begin{theorem}
	\label{SW+ theorem}
	For any type of adversaries, with probability at least $1-\delta$, the weak regret of {\algC} is bounded by  $O(T^{\frac{2}{3}})$. In particular, choosing
	$	\tau = B_{\tau} T^{\frac{1}{3}} \in \left[1, T\right],$
	$	\gamma = B_{\gamma} T^{-\frac{1}{3}} \in (0, \frac{1}{2}),$
	$	\beta = B_{\beta} T^{-\frac{1}{3}} \in (0, 1),$
	$	\text{and }  \eta  = \frac{B_{\gamma}}{2lC} T^{-\frac{1}{3}}$,
	where $ B_{\tau}$, $B_{\gamma}$, and $B_{\beta}$ are constants, 
	we have
	\begin{align}
		\label{eq: bound Gbest reward confidence}
		R_{\algCabr}
		\leq 
		\left(
		2B_{\gamma} + B_{\beta}K  	
		+  B_{\tau}
		\left(
		\frac{ l \ln \frac{K}{\delta}}{B_{\beta}} + \frac{\ln S}{B_{\eta}}
		\right)	
		+ \frac{1}{B_{\tau}}
		\right) T^{\frac{2}{3}}.
	\end{align}
\end{theorem}
\begin{proof}
	Similar to the proof of Theorem~\ref{thm3}, we have
	\begin{align*}
	R_{\algBabr}
	=~& G_{best} - L_{best} - G_{\algBabr} + L_{\algBabr}
	\\
	\leq~&
	G_{best} - G_{\algBabr} + J
	\\
	=~&
	\tau\bar{\mathsf{G}}_{best} 
	- \tau \bar{\mathsf{G}}_{\algBabr}
	+ \frac{T}{\tau}
	\\
	\leq~&
	\tau
	\left(
	\gamma J + 2\eta lC J
	+	
	\frac{l}{\beta} \ln \frac{K}{\delta}
	+\frac{\ln S }{\eta}
	+ \beta KJ
	\right)	
	+ \frac{T}{\tau},
	\\
	\leq~&	
	\left(
	\gamma + 2\eta lC + \beta K	
	\right) T	
	+
	\tau
	\left(
	\frac{l}{\beta} \ln \frac{K}{\delta}
	+\frac{\ln S }{\eta}
	\right)	
	+ 
	\frac{T}{\tau},
	\stepcounter{equation}\tag{\theequation}\label{eq: + eq}
	\end{align*}
	with probability at least $1- \delta$.
	The last inequality follows from Lemma \ref{SW+ reward}. Plugging in the value of parameters finishes the proof. 
\end{proof}

We now provide an example choice of parameters to reach a specific bound. 
\begin{corollary}\label{SW+ special}
	For any type of adversaries, under the condition of 	
	\begin{equation*}
	T \geq \max \left\{B^2, \frac{8 \left(l C \ln S\right)^{3/2}}{B}, \frac{\left(\frac{l}{K} \ln \frac{K}{\delta} \right)^{3/2}}{B} \right\},	
	\end{equation*}
	using parameters
	$\tau = B^{-\frac{2}{3}} T^{\frac{1}{3}},$
	$\gamma = \sqrt{lC \ln S} \cdot B^{-\frac{1}{3}} T^{-\frac{1}{3}},$
	$\beta =  \sqrt{\frac{l}{K} \ln \frac{K}{\delta}} \cdot B^{-\frac{1}{3}} T^{-\frac{1}{3}},$
	$\text{and } \eta  = \sqrt{\frac{\ln S}{4lC}} \cdot B^{-\frac{1}{3}} T^{-\frac{1}{3}}$,
	where $B = 4\sqrt{lC \ln S} + 2\sqrt{lK\ln\frac{K}{\delta}}$, we have
	\begin{equation}
	R_{\algCabr}
	\leq
	2 \left(4\sqrt{lC\ln S} + 2\sqrt{lK\ln\frac{K}{\delta}}\right)^{\frac{2}{3}} T^{\frac{2}{3}},
	\end{equation}
	with probability at least $1-\delta$.
\end{corollary}
\begin{proof}
	Substituting the parameters in \eqref{eq: bound Gbest reward confidence}, we have the immediate result.
\end{proof}

\noindent \textit{Remark.}
Note that it is not guaranteed that {\algC} always outperforms {\algB}. The improvement over {\algB} is the fact that {\algC} guarantees the actual weak regret to be bounded with any predefined confidence level.
	
	
%
%
\section{Performance Evaluation}\label{simulation}
\begin{figure*}[!htb]
		\begin{minipage}[t]{0.32\textwidth}
		\centering
		\includegraphics[width=.99\textwidth]{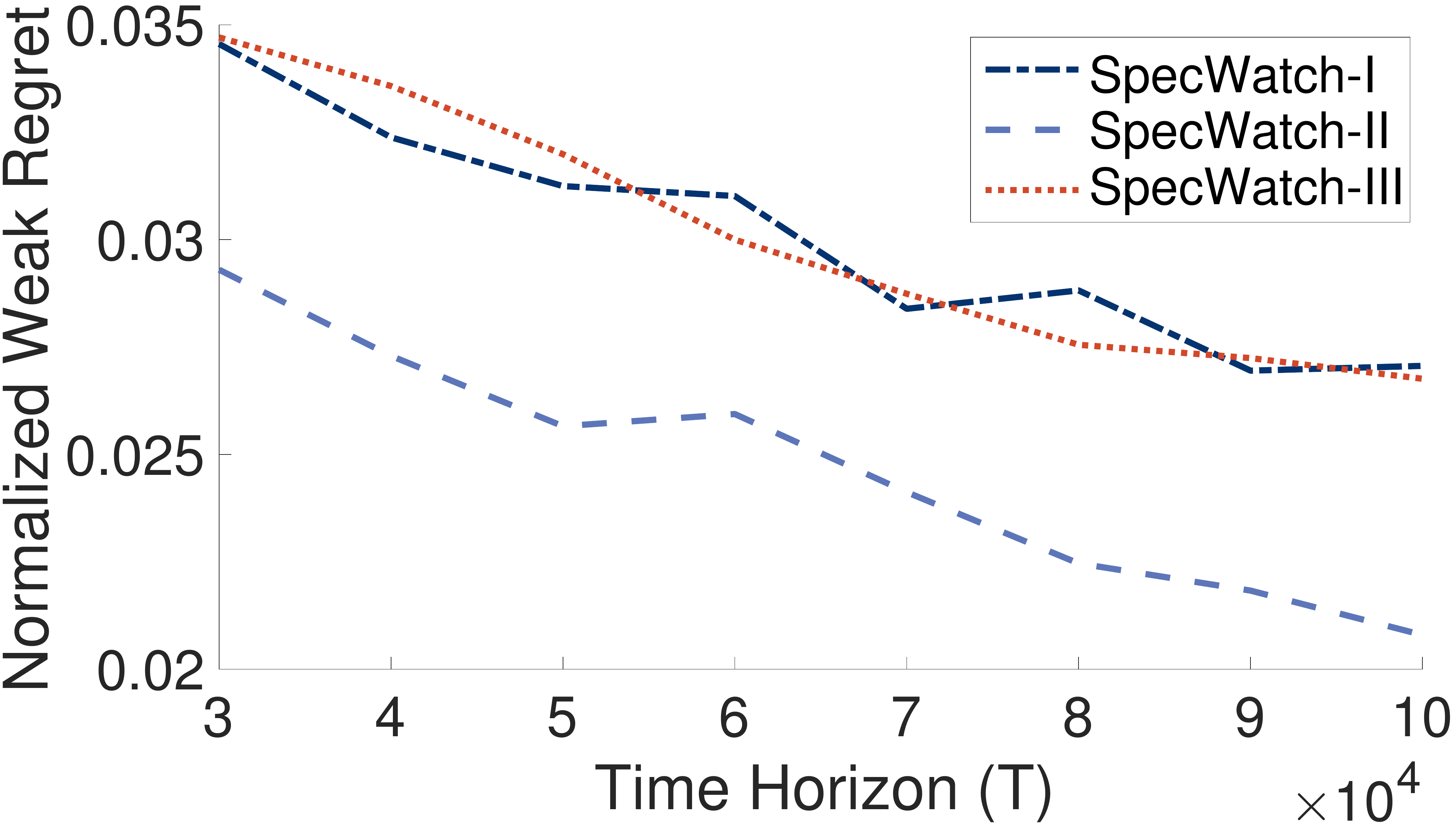}	
		\caption{Convergence of weak regrets} 
		\label{fig:ttoregret}
	\end{minipage}
\begin{minipage}[t]{0.32\textwidth}
		\centering
		\includegraphics[width=.99\textwidth]{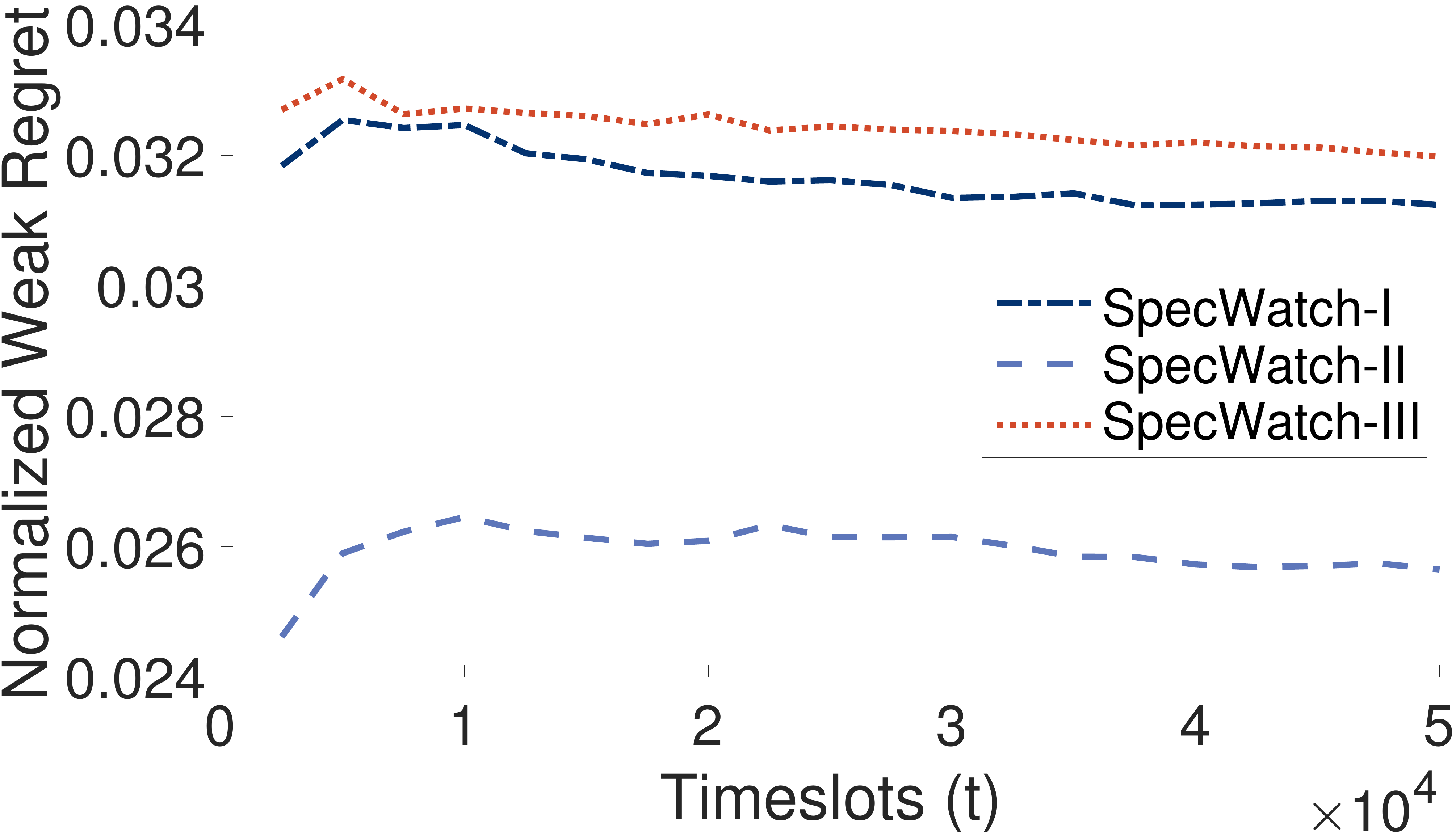}	
		\caption{Algorithm Comparison}
		\label{fig:compare}
	\end{minipage}
\begin{minipage}[t]{0.32\textwidth}
		\centering
		\includegraphics[width=.99\textwidth]{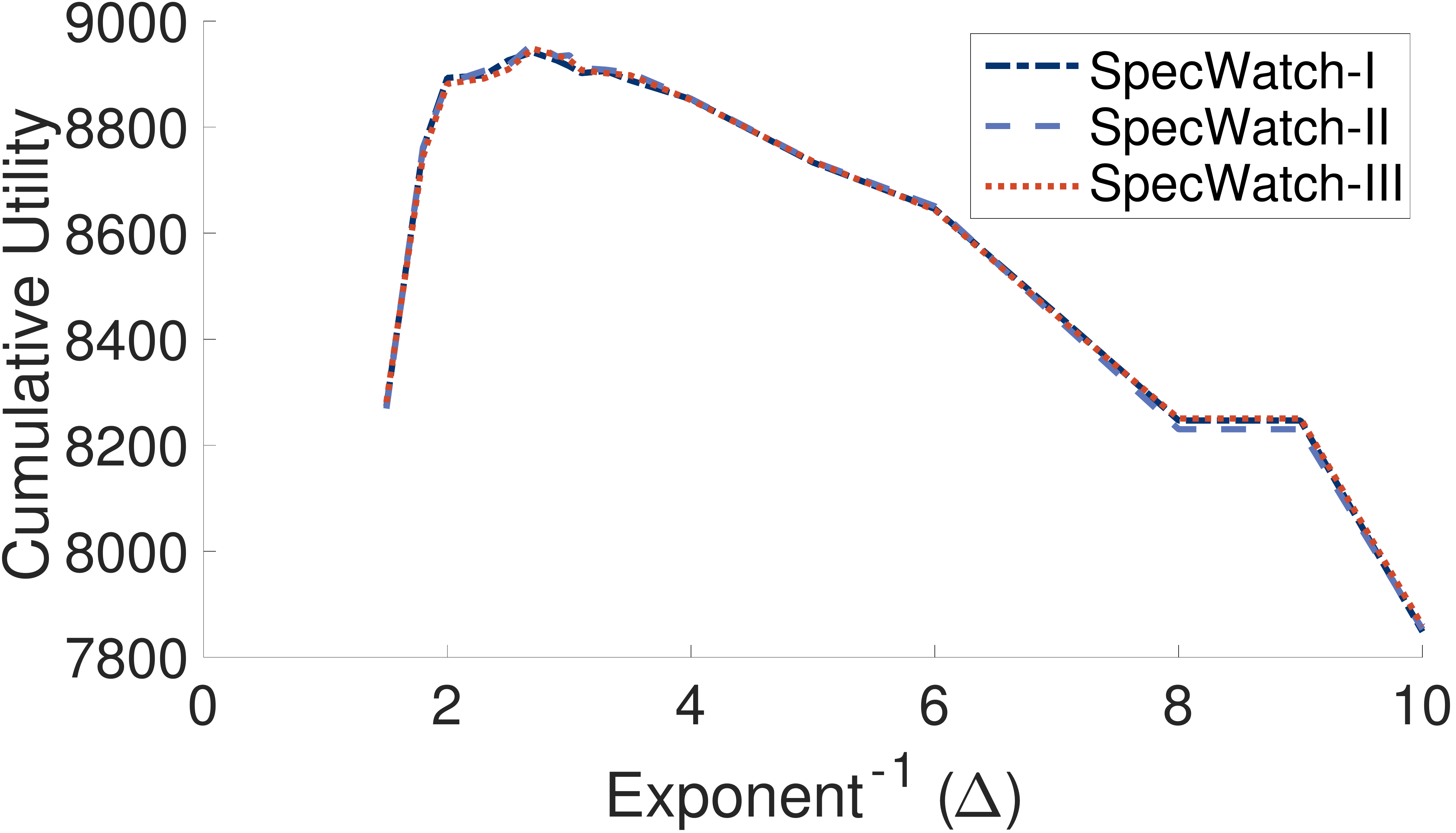}	
		\caption{Impact of batch size on cumulative utility}
		\label{fig:tau}
	\end{minipage}
\begin{minipage}[t]{0.1\textwidth}	
			\quad	
		\end{minipage}
\begin{minipage}[t]{0.32\textwidth}
		\centering
		\includegraphics[width=.99\textwidth]{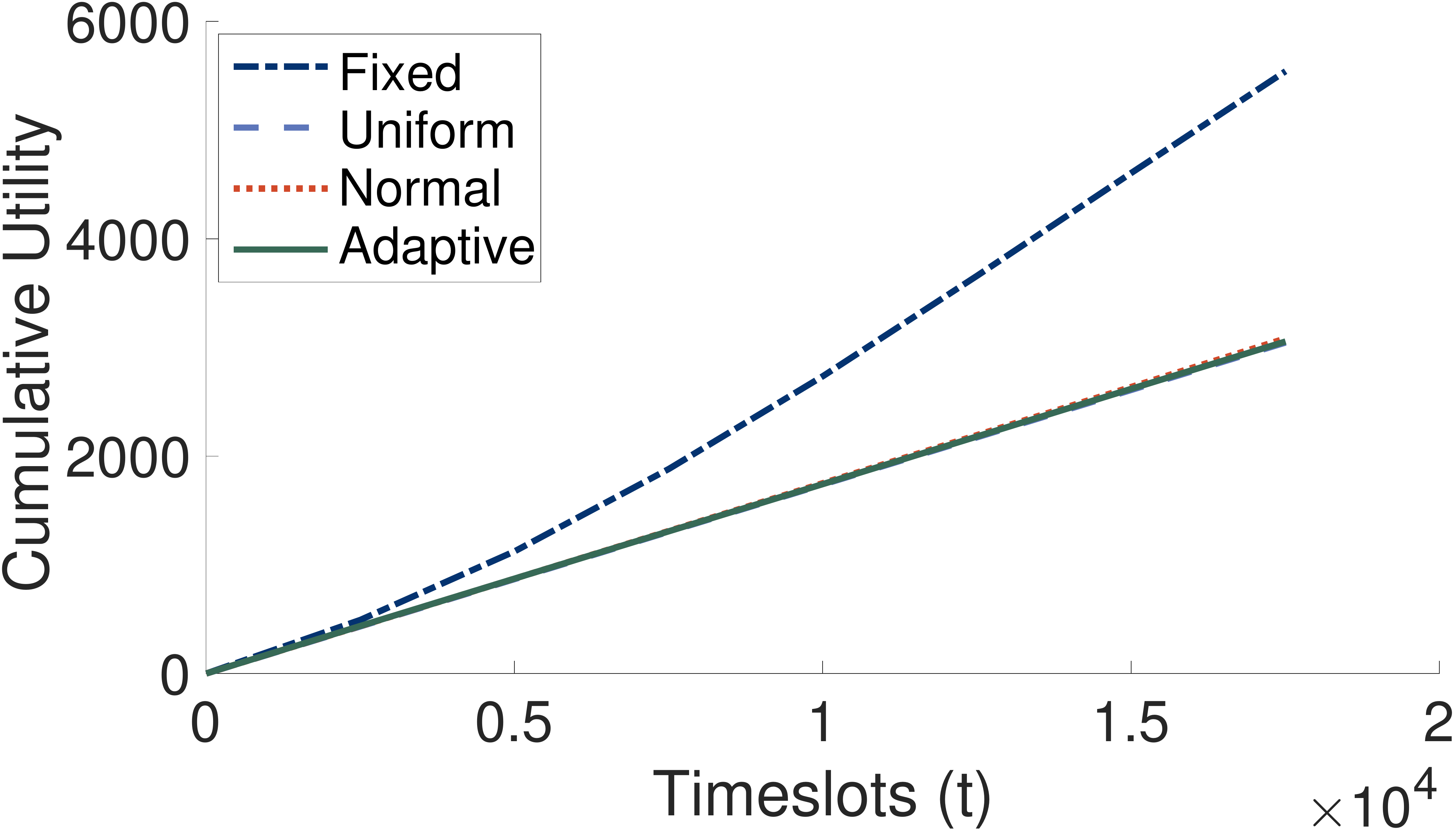}	
		\caption{Cumulative utility  under different adversary settings}
		\label{fig:typeU3abs}
	\end{minipage}	
\begin{minipage}[t]{0.32\textwidth}
		\centering
		\includegraphics[width=.99\textwidth]{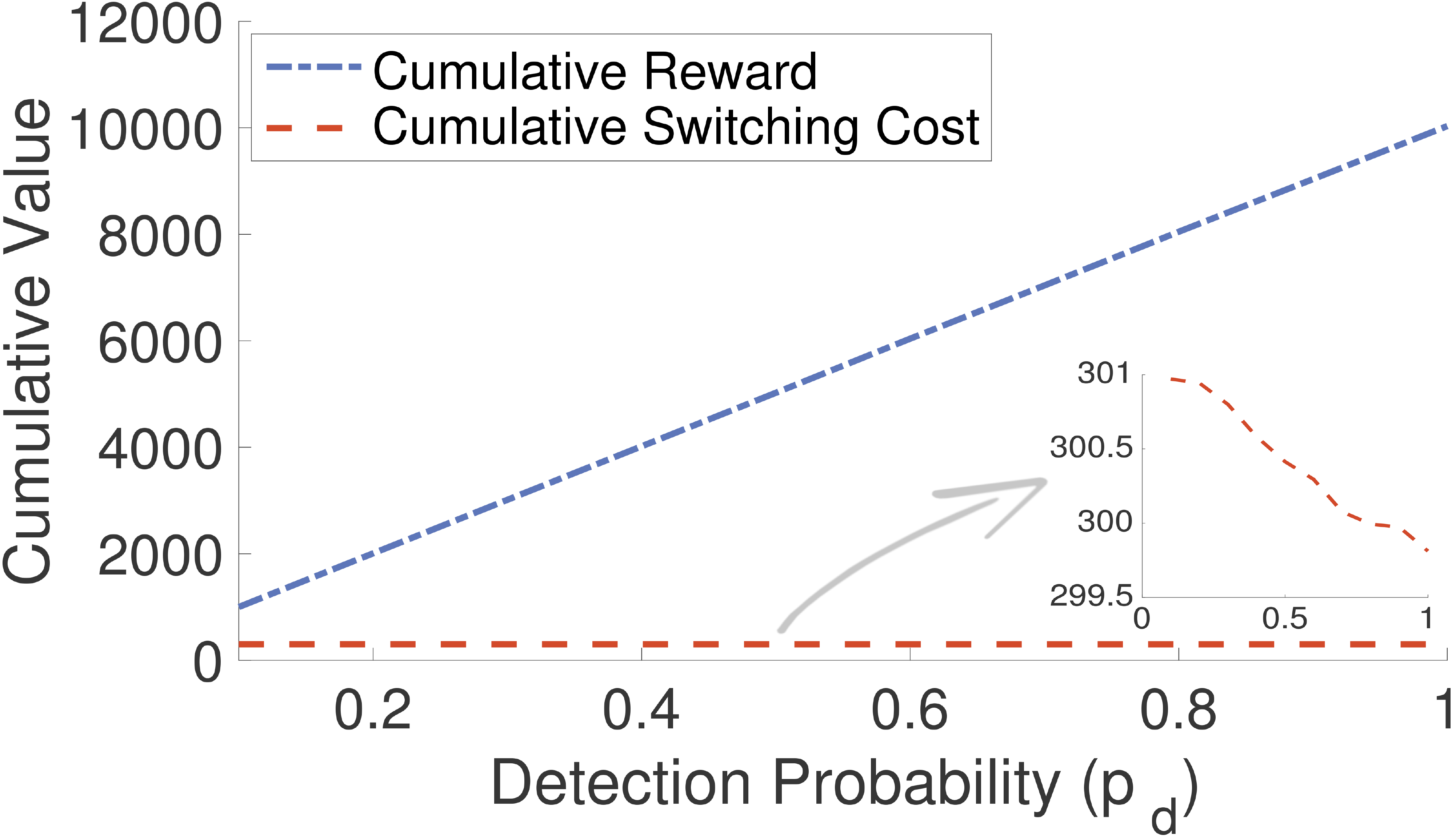} 
		\caption{Impact of detection probability on cumulative reward and switching cost} 
		\label{fig:pd}
	\end{minipage}	
\begin{minipage}[t]{0.1\textwidth}	
	\quad	
\end{minipage}
\begin{minipage}[t]{0.32\textwidth}
		\centering
		\includegraphics[width=.99\textwidth]{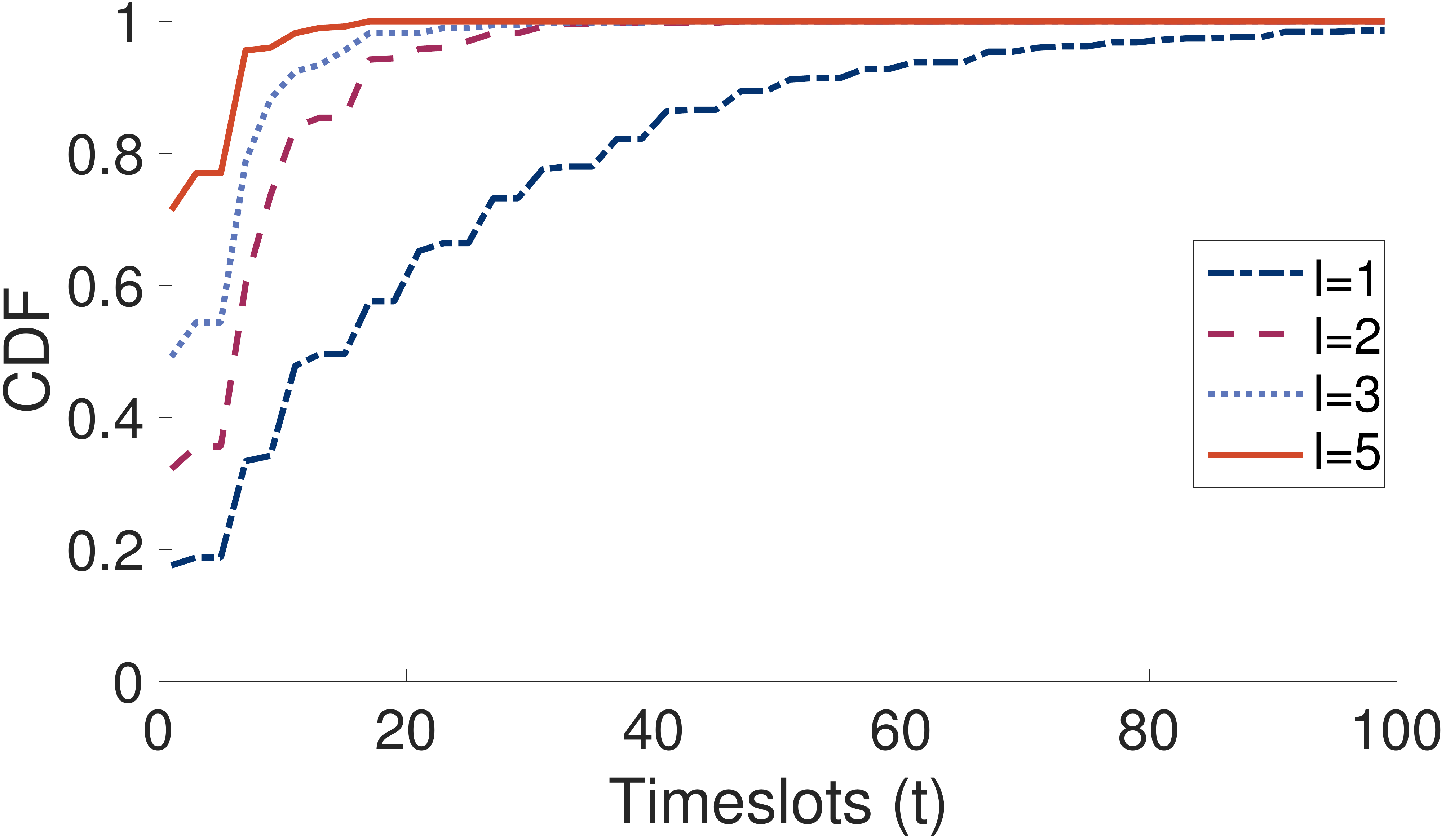}	
		\caption{Number of timeslots to detect the first misuse with different number of radios}
		\label{fig:CDF_l}
	\end{minipage}	
\begin{minipage}[t]{0.32\textwidth}
		\centering
		\includegraphics[width=.99\textwidth]{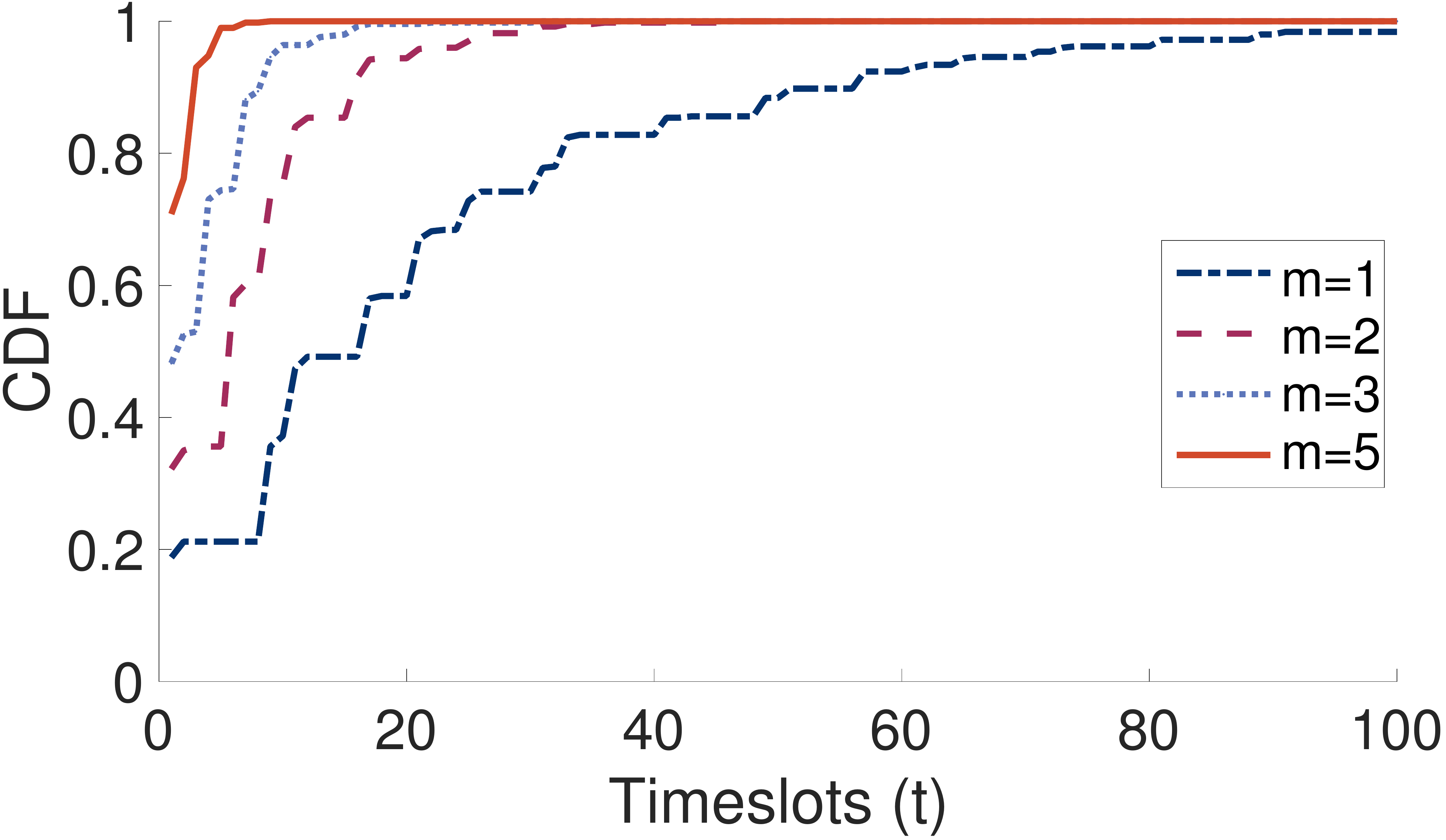}	
		\caption{Number of timeslots to detect the first misuse with different number of MUs} 
		\label{fig:CDF_m}
	\end{minipage}
\begin{minipage}[t]{0.32\textwidth}
		\centering
		\includegraphics[width=.99\textwidth]{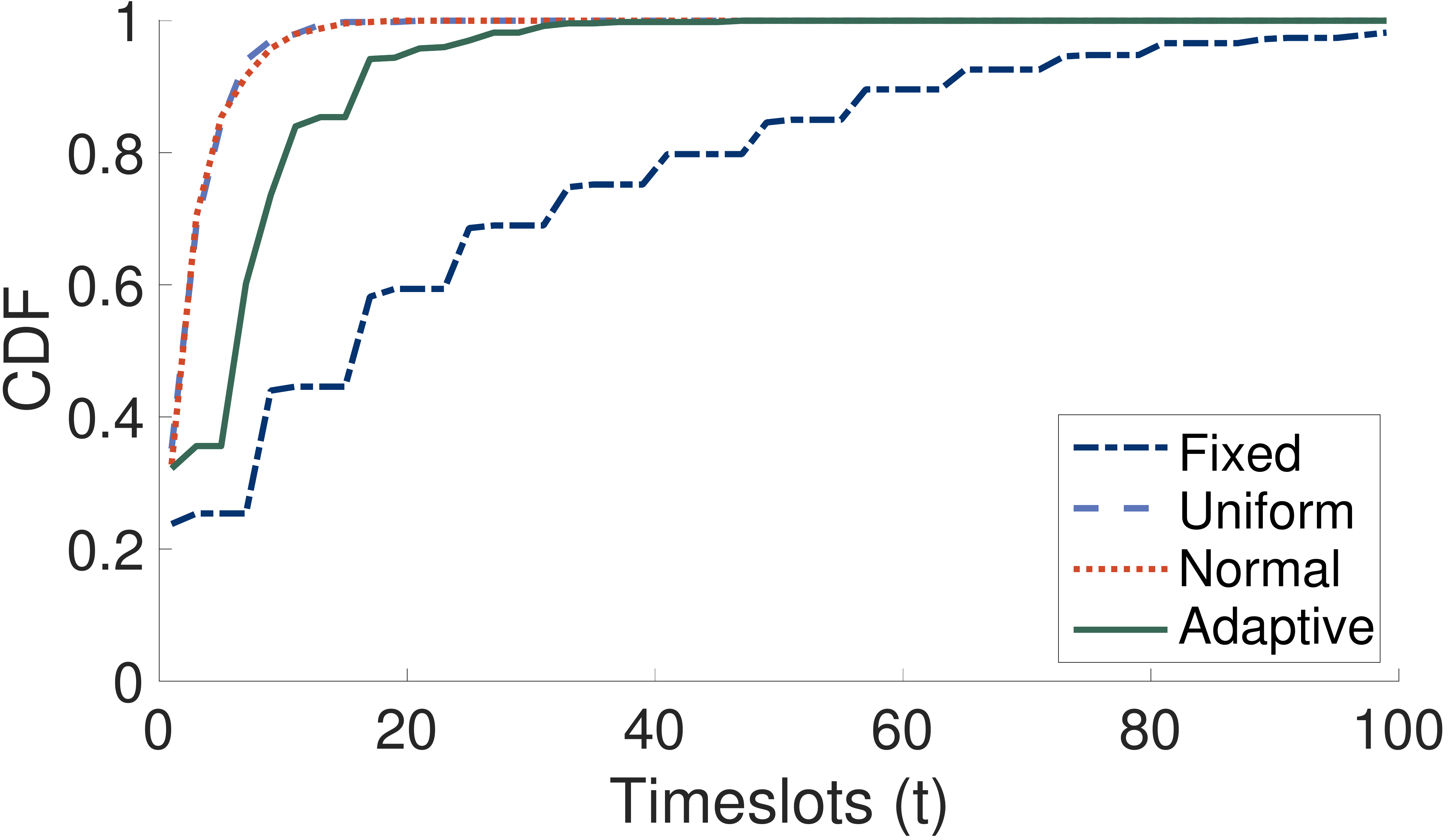}	
		\caption{Number of timeslots to detect the first misuse under different adversary settings} 
		\label{fig:CDF_MUtype}
	\end{minipage}
\begin{minipage}[t]{0.1\textwidth}	
	\quad	
\end{minipage}
\end{figure*}

We conduct extensive simulations to demonstrate the performance of our proposed online spectrum monitoring algorithms, {\algA}, {\algB}, and {\algC}. We first show the convergence of normalized weak regrets of all four algorithms and then study and compare their performances under different adversary settings. We also demonstrate the impact of the detection probability, the number of radios, the number of MUs, and adversary settings, on the algorithm performance. 

In the simulation setting, we consider $K=10$ channels, and we deploy a monitor with $l=2$ radios. We set the unit reward of successfully detecting on a single channel to be $r=0.3$ and the unit switching cost of tuning one radio to be $c=0.03$. If not specified, the detection probability of each radio is set to be $p_d=0.9$ as it is the recommended detection accuracy in consistent with~\cite{choi2013time}. The parameters of all algorithms are chosen as in the corollaries. If the monitor uses {\algC}, we set $\delta=0.5$ so that the weak regret is relatively small with an acceptable confidence level.

We assume there are $m=2$ MUs attacking channels either obliviously or adaptively.  Specifically, we consider four adversary settings, 
\begin{itemize}
	\item Fixed adversary (\textit{Fixed}): Each MU selects a fixed channel and never switches throughout the time horizon $T$.
	
	\item Uniform adversary (\textit{Uniform}): In every timeslot, each MU selects a channel uniformly at random.
	
	\item Normal adversary (\textit{Normal}): In every timeslot, each MU selects a channel following the same normal distribution.
	
	\item Adaptive adversary (\textit{Adaptive}): Each MU adopts modified {\algA}, where the actual channel reward is $r$ if the MU is not captured on that channel, and 0 otherwise.
	
\end{itemize}

The simulation results are shown below, and each of them is averaged over 100 trials.

\noindent \textbf{Weak Regret.} Fig.~\ref{fig:ttoregret} shows the normalized weak regrets of all algorithms decrease with time horizon $T$, which supports our theoretical analysis that the normalized weak regret converges to 0 as $T \to \infty$. In all following simulations, we fix the time horizon to be $T=50000$. Here we only show the result with adaptive adversary since in other adversary settings, the results are similar. Note that {\algB} outperforms other algorithms as shown in the figure, but it only means that {\algB}'s way to trade-off between exploration and exploitation is more appropriate under our current simulation setting.

Fig.~\ref{fig:compare} plots how the normalized weak regret decreases with timeslots under adaptive adversary. At the beginning, there is apparent fluctuation of the normalized weak regret. As time goes by, the monitor and the adaptive adversary enter a relatively stable stage, but we can still see the decreasing trend of the normalized weak regret, which indicates our algorithms are learning from monitoring history to make smart decisions.

\noindent \textbf{Impact of Algorithm Parameters.} Among all parameters of the four algorithms, the most important one is the batch size $\tau$, which controls the trade-off between cumulative reward and cumulative switching cost. As shown in Fig. \ref{fig:tau}, we conducted simulations where the batch size $\tau$ was set to be exactly $T^{1/\Delta}$ 
and plotted how the cumulative utility changes with $\Delta$, under the adaptive adversary.
It is shown that all algorithms achieve highest cumulative utility ratio when $\tau$ is around $T^{1/3}$, which is in consistency with our theoretical analysis. The performances of all algorithms are almost the same because they are designed based on the same framework and using same way to trade-off between rewards and switching costs.

\noindent \textbf{Cumulative Utility.} Fig.~\ref{fig:typeU3abs} 
plots the actual utilities gained by {\algC}. The figures for the other three algorithms are very similar and thus omitted.
We observe that the cumulative utilities under fixed adversary greatly exceed the other three settings, and the other three settings have similar results.

\noindent \textbf{Impact of System Parameters and Adversary Settings.}  In simulations, we fix the time horizon to be $T=50000$. Since the impacts on all algorithms are similar, we only present results of {\algC}.

Fig.~\ref{fig:pd} shows the impact of detection probability $p_d$ on the cumulative rewards and the cumulative switching cost. 

As expected,  the cumulative reward grows with decreasing slope as the detection probability increases.
The cumulative switching cost, however, has a decreasing trend. This is because the larger the detection probability, the more accurate for the monitor to evaluate each strategy; thus the best strategy is revealed more quickly, avoiding unnecessary switches and reducing cumulative switching cost.

We also study the impact the number of radios $l$, the number of MUs $m$, and the types of adversary on the performance of our algorithms. Fig.~\ref{fig:CDF_l}, Fig.~\ref{fig:CDF_m}, and Fig.~\ref{fig:CDF_MUtype} illustrate the cumulative distribution function (CDF) of expected number of timeslots to detect the first misuse. In general,  more radios or more MUs  make it sooner for the monitor to detect successfully. In Fig.~\ref{fig:CDF_MUtype}, the monitor takes the longest time to detect the first misuse under fixed adversary setting, which is because the monitor sticks to the same strategy for the whole batch to prevent switching costs. If the monitor does not choose the channels attacked by MUs at the first timeslot in a batch, it will not detect misuse for the following $\tau - 1$ timeslots. As a result, it takes longer time for the monitor to detect the first misuse under fixed adversary setting.


	
	\section{Conclusion}
In this paper, we studied the adversarial {\problemName} with unknown statistics by formulating it as an adversarial multi-armed bandit problem with switching costs (MAB-SC). To solve this problem, we proposed an online spectrum monitoring framework named {\alg} and designed two effective online algorithms, {\algsfull}. We rigorously proved that their weak regrets are bounded by $O(T^{2/3})$, which matches the lower bound of the general MAB-SC problem. Thus, they are asymptotically optimal. Moreover, our algorithms can guarantee the proved performance under any adversary setting and are independent of the underlying misuse detection technique.
	
	



\begin{thebibliography}{10}
\providecommand{\url}[1]{#1}
\csname url@samestyle\endcsname
\providecommand{\newblock}{\relax}
\providecommand{\bibinfo}[2]{#2}
\providecommand{\BIBentrySTDinterwordspacing}{\spaceskip=0pt\relax}
\providecommand{\BIBentryALTinterwordstretchfactor}{4}
\providecommand{\BIBentryALTinterwordspacing}{\spaceskip=\fontdimen2\font plus
\BIBentryALTinterwordstretchfactor\fontdimen3\font minus
  \fontdimen4\font\relax}
\providecommand{\BIBforeignlanguage}[2]{{%
\expandafter\ifx\csname l@#1\endcsname\relax
\typeout{** WARNING: IEEEtran.bst: No hyphenation pattern has been}%
\typeout{** loaded for the language `#1'. Using the pattern for}%
\typeout{** the default language instead.}%
\else
\language=\csname l@#1\endcsname
\fi
#2}}
\providecommand{\BIBdecl}{\relax}
\BIBdecl

\bibitem{force2002report}
F.~S. P.~T. Force, ``Report of the spectrum efficiency working group,'' Tech.
  Rep., 2002.

\bibitem{zhao2007survey}
Q.~Zhao and B.~M. Sadler, ``A survey of dynamic spectrum access,'' \emph{IEEE
  Signal Process. Mag.}, vol.~24, no.~3, pp. 79--89, 2007.

\bibitem{santivanez2006opportunistic}
C.~Santivanez, R.~Ramanathan, C.~Partridge, R.~Krishnan, M.~Condell, and
  S.~Polit, ``Opportunistic spectrum access: Challenges, architecture,
  protocols,'' in \emph{Proc. of WICON}, 2006, p.~13.

\bibitem{anandkumar2010opportunistic}
A.~Anandkumar, N.~Michael, and A.~Tang, ``Opportunistic spectrum access with
  multiple users: learning under competition,'' in \emph{Proc. of INFOCOM},
  2010, pp. 1--9.

\bibitem{tekin2011online}
C.~Tekin and M.~Liu, ``Online learning in opportunistic spectrum access: A
  restless bandit approach,'' in \emph{Proc. of INFOCOM}, 2011, pp. 2462--2470.

\bibitem{xu2012opportunistic}
Y.~Xu, J.~Wang, Q.~Wu, A.~Anpalagan, and Y.-D. Yao, ``Opportunistic spectrum
  access in unknown dynamic environment: A game-theoretic stochastic learning
  solution,'' \emph{{IEEE} Trans. Wireless Commun.}, vol.~11, no.~4, pp.
  1380--1391, 2012.

\bibitem{altrad2014opportunistic}
O.~Altrad, S.~Muhaidat, A.~Al-Dweik, A.~Shami, and P.~D. Yoo, ``Opportunistic
  spectrum access in cognitive radio networks under imperfect spectrum
  sensing,'' \emph{{IEEE} Trans. Veh. Technol.}, vol.~63, no.~2, pp. 920--925,
  2014.

\bibitem{yadav2015opportunistic}
R.~N. Yadav, R.~Misra, U.~Gupta, and S.~Bhagat, ``Opportunistic spectrum access
  in cr network in licensed and unlicensed channels,'' in \emph{Proc. of
  ICDCN}, 2015, p.~39.

\bibitem{wang2016jamming}
Q.~Wang, K.~Ren, P.~Ning, and S.~Hu, ``Jamming-resistant multiradio
  multichannel opportunistic spectrum access in cognitive radio networks,''
  \emph{{IEEE} Trans. Veh. Technol.}, vol.~65, no.~10, pp. 8331--8344, 2016.

\bibitem{tsiropoulos2016radio}
G.~I. Tsiropoulos, O.~A. Dobre, M.~H. Ahmed, and K.~E. Baddour, ``Radio
  resource allocation techniques for efficient spectrum access in cognitive
  radio networks,'' \emph{IEEE Commun. Surveys \& Tutorials}, vol.~18, no.~1,
  pp. 824--847, 2016.

\bibitem{wang2011anti}
Q.~Wang, K.~Ren, and P.~Ning, ``Anti-jamming communication in cognitive radio
  networks with unknown channel statistics,'' in \emph{Proc. of ICNP}, 2011,
  pp. 393--402.

\bibitem{yang2012enforcing}
L.~Yang, Z.~Zhang, B.~Y. Zhao, C.~Kruegel, and H.~Zheng, ``Enforcing dynamic
  spectrum access with spectrum permits,'' in \emph{Proc. of MobiHoc}, 2012,
  pp. 195--204.

\bibitem{weiss2013enforcement}
M.~B. Weiss, M.~Altamimi, and M.~McHenry, ``Enforcement and spectrum sharing: A
  case study of the 1695--1710 mhz band,'' in \emph{Proc. of CROWNCOM}, 2013,
  pp. 7--12.

\bibitem{sorrells2011anomalous}
C.~Sorrells, P.~Potier, L.~Qian, and X.~Li, ``Anomalous spectrum usage attack
  detection in cognitive radio wireless networks,'' in \emph{Proc. of HST},
  2011, pp. 384--389.

\bibitem{kumar2016frequency}
G.~P. Kumar and D.~K. Reddy, ``Frequency domain techniques for void spectrum
  detection in cognitive radio network for emulation attack prevention,'' in
  \emph{Proc. of CCPCT}, 2016, pp. 1--6.

\bibitem{atia2008spectrum}
G.~Atia, A.~Sahai, and V.~Saligrama, ``Spectrum enforcement and liability
  assignment in cognitive radio systems,'' in \emph{Proc. of DySPAN}, 2008, pp.
  1--12.

\bibitem{kyasanur2003detection}
P.~Kyasanur and N.~H. Vaidya, ``Detection and handling of mac layer misbehavior
  in wireless networks,'' in \emph{Proc. of DSN}, 2003, pp. 173--182.

\bibitem{zhang2013countering}
Y.~Zhang and L.~Lazos, ``Countering selfish misbehavior in multi-channel mac
  protocols,'' in \emph{Proc. of INFOCOM}, 2013, pp. 2787--2795.

\bibitem{liu2012detecting}
S.~Liu, L.~J. Greenstein, W.~Trappe, and Y.~Chen, ``Detecting anomalous
  spectrum usage in dynamic spectrum access networks,'' \emph{Ad Hoc Netw.},
  vol.~10, no.~5, pp. 831--844, 2012.

\bibitem{tang2013selfish}
J.~Tang and Y.~Cheng, ``Selfish misbehavior detection in 802.11 based wireless
  networks: An adaptive approach based on markov decision process,'' in
  \emph{Proc. of INFOCOM}, 2013, pp. 1357--1365.

\bibitem{Jin2015specguard}
X.~Jin, J.~Sun, R.~Zhang, Y.~Zhang, and C.~Zhang, ``{SpecGuard}: Spectrum
  misuse detection in dynamic spectrum access systems.'' in \emph{Proc. of
  INFOCOM}, 2015, pp. 172--180.

\bibitem{dutta2016see}
A.~Dutta and M.~Chiang, ```see something, say something' crowdsourced
  enforcement of spectrum policies,'' \emph{{IEEE} Trans. Wireless Commun.},
  vol.~15, no.~1, pp. 67--80, 2016.

\bibitem{baltiiski2016long}
P.~Baltiiski, I.~Iliev, B.~Kehaiov, V.~Poulkov, and T.~Cooklev, ``Long-term
  spectrum monitoring with big data analysis and machine learning for
  cloud-based radio access networks,'' \emph{Wireless Personal Commun.},
  vol.~87, no.~3, pp. 815--835, 2016.

\bibitem{auer2002nonstochastic}
P.~Auer, N.~Cesa-Bianchi, Y.~Freund, and R.~E. Schapire, ``The nonstochastic
  multiarmed bandit problem,'' \emph{SIAM J. on Computing}, vol.~32, no.~1, pp.
  48--77, 2002.

\bibitem{dekel2014bandits}
O.~Dekel, J.~Ding, T.~Koren, and Y.~Peres, ``Bandits with switching costs:
  {$T$}$^{2/3}$ regret,'' in \emph{Proc. of STOC}, 2014, pp. 459--467.

\bibitem{shin2009optimal}
D.-H. Shin and S.~Bagchi, ``Optimal monitoring in multi-channel multi-radio
  wireless mesh networks,'' in \emph{Proc. of MobiHoc}, 2009, pp. 229--238.

\bibitem{shin2013toward}
D.-H. Shin, S.~Bagchi, and C.-C. Wang, ``Toward optimal sniffer-channel
  assignment for reliable monitoring in multi-channel wireless networks,'' in
  \emph{Proc. of SECON}, 2013, pp. 203--211.

\bibitem{nguyen2014quality}
H.~Nguyen, G.~Scalosub, and R.~Zheng, ``On quality of monitoring for
  multichannel wireless infrastructure networks,'' \emph{{IEEE} Trans. Mobile
  Comput.}, vol.~13, no.~3, pp. 664--677, 2014.

\bibitem{chen2014efficient}
S.~Chen, K.~Zeng, and P.~Mohapatra, ``Efficient data capturing for network
  forensics in cognitive radio networks,'' \emph{{IEEE/ACM} Trans. Netw.},
  vol.~22, no.~6, pp. 1988--2000, 2014.

\bibitem{shin2012distributed}
D.-H. Shin, S.~Bagchi, and C.-C. Wang, ``Distributed online channel assignment
  toward optimal monitoring in multi-channel wireless networks,'' in
  \emph{Proc. of INFOCOM}, 2012, pp. 2626--2630.

\bibitem{yan2014specmonitor}
Q.~Yan, M.~Li, F.~Chen, T.~Jiang, W.~Lou, Y.~T. Hou, and C.-T. Lu,
  ``{SpecMonitor}: Towards efficient passive traffic monitoring for cognitive
  radio networks,'' \emph{{IEEE} Trans. Wireless Commun.}, vol.~13, no.~10, pp.
  5893--5905, 2014.

\bibitem{arora2011sequential}
P.~Arora, C.~Szepesv{\'a}ri, and R.~Zheng, ``Sequential learning for optimal
  monitoring of multi-channel wireless networks,'' in \emph{Proc. of INFOCOM},
  2011, pp. 1152--1160.

\bibitem{zheng2014approximate}
R.~Zheng, T.~Le, and Z.~Han, ``Approximate online learning algorithms for
  optimal monitoring in multi-channel wireless networks,'' \emph{{IEEE} Trans.
  Wireless Commun.}, vol.~13, no.~2, pp. 1023--1033, 2014.

\bibitem{le2014sequential}
T.~Le, C.~Szepesvari, and R.~Zheng, ``Sequential learning for multi-channel
  wireless network monitoring with channel switching costs,'' \emph{{IEEE}
  Trans. Signal Process.}, vol.~62, no.~22, pp. 5919--5929, 2014.

\bibitem{yi2012secondary}
S.~Yi, K.~Zeng, and J.~Xu, ``Secondary user monitoring in unslotted cognitive
  radio networks with unknown models,'' in \emph{Proc. of WASA}, 2012, pp.
  648--659.

\bibitem{xu2014secondary}
J.~Xu, Q.~Wang, R.~Jin, K.~Zeng, and M.~Liu, ``Secondary user data capturing
  for cognitive radio network forensics under capturing uncertainty,'' in
  \emph{Proc. of MILCOM}, 2014, pp. 935--941.

\bibitem{auer2002finite}
P.~Auer, N.~Cesa-Bianchi, and P.~Fischer, ``Finite-time analysis of the
  multiarmed bandit problem,'' \emph{Mach. Learn.}, vol.~47, no. 2-3, pp.
  235--256, 2002.

\bibitem{robbins1952some}
H.~Robbins, ``Some aspects of the sequential design of experiments,''
  \emph{Bull. Amer. Math. Soc.}, vol.~58, no.~5, pp. 527--535, 1952.

\bibitem{lai1985asymptotically}
T.~L. Lai and H.~Robbins, ``Asymptotically efficient adaptive allocation
  rules,'' \emph{Adv. Appl. Math.}, vol.~6, no.~1, pp. 4--22, 1985.

\bibitem{bubeck2012regret}
S.~Bubeck, N.~Cesa-Bianchi \emph{et~al.}, ``Regret analysis of stochastic and
  nonstochastic multi-armed bandit problems,'' \emph{Foundations and
  Trends{\textregistered} in Mach. Learn.}, vol.~5, no.~1, pp. 1--122, 2012.

\bibitem{auer2016algorithm}
P.~Auer and C.-K. Chiang, ``An algorithm with nearly optimal pseudo-regret for
  both stochastic and adversarial bandits,'' in \emph{Proc. of COLT}, 2016, pp.
  116--120.

\bibitem{lin2011Opportunistic}
L.~Chen, S.~Iellamo, and M.~Coupechoux, ``Opportunistic spectrum access with
  channel switching cost for cognitive radio networks,'' in \emph{Proc. of
  ICC}, 2011, pp. 1--5.

\bibitem{IEEE80211af}
``{IEEE} standard for information technology--telecommunications and
  information exchange between systems local and metropolitan area
  networks--specific requirements - part 11: Wireless {LAN} medium access
  control ({MAC}) and physical layer ({PHY}) specifications,'' \emph{IEEE Std
  802.11-2016}, pp. 1--3534, 2016.

\bibitem{wang2016learning}
Q.~Wang and M.~Liu, ``Learning in hide-and-seek,'' \emph{{IEEE/ACM} Trans.
  Netw.}, vol.~24, no.~2, pp. 1279--1292, 2016.

\bibitem{arora2012online}
R.~Arora, O.~Dekel, and A.~Tewari, ``Online bandit learning against an adaptive
  adversary: from regret to policy regret,'' in \emph{Proc. of ICML}, 2012, pp.
  1503--1510.

\bibitem{li2016specwatch}
M.~Li, D.~Yang, M.~Li, J.~Lin, and J.~Tang, ``{SpecWatch}: Adversarial spectrum
  usage monitoring in {CRN}s with unknown statistics,'' in \emph{Proc. of
  INFOCOM}, 2016, pp. 1--9.

\bibitem{gyorgy2007line}
A.~Gyorgy, T.~Linder, G.~Lugosi, and G.~Ottucsak, ``The on-line shortest path
  problem under partial monitoring,'' \emph{J. Mach. Learn. Res.}, vol.~8, pp.
  2369--2403, 2007.

\bibitem{choi2013time}
J.-K. Choi and S.-J. Yoo, ``Time-constrained detection probability and sensing
  parameter optimization in cognitive radio networks,'' \emph{EURASIP J.
  Wireless Commun. Netw.}, vol. 2013, no.~1, pp. 1--12, 2013.

\end{thebibliography}
	
	
\end{document}